\documentclass[11pt]{amsart}

\usepackage{amsmath, amsthm, amssymb, amsfonts, enumerate,verbatim}
\usepackage[colorlinks=true,linkcolor=blue,urlcolor=blue]{hyperref}
\usepackage{dsfont}
\usepackage{color}
\usepackage{geometry}
\usepackage{todonotes}
\usepackage{bbm}

\newtheorem{theorem}{Theorem}[section]
\newtheorem{assumption}[theorem]{Assumption}
\newtheorem{lemma}[theorem]{Lemma}
\newtheorem{proposition}[theorem]{Proposition}
\newtheorem{corollary}[theorem]{Corollary}

\theoremstyle{definition}
\newtheorem{remark}[theorem]{Remark}

\title[Optimal Dividends with Capital Injections]{An Optimal Dividend Problem with Capital Injections\\ over a Finite Horizon}
\author[Ferrari, Schuhmann]{Giorgio Ferrari, Patrick Schuhmann}
\address{G.~Ferrari: Center for Mathematical Economics (IMW), Bielefeld University, Universit\"atsstrasse 25, 33615, Bielefeld, Germany}
\email{\href{mailto:giorgio.ferrari@uni-bielefeld.de}{giorgio.ferrari@uni-bielefeld.de}}
\address{P.~Schuhmann: Center for Mathematical Economics (IMW), Bielefeld University, Universit\"atsstrasse 25, 33615, Bielefeld, Germany}
\email{\href{mailto:patrick.schuhmann@uni-bielefeld.de}{patrick.schuhmann@uni-bielefeld.de}}
\date{\today}

\numberwithin{equation}{section}

\begin{document}

\begin{abstract} 
In this paper we propose and solve an optimal dividend problem with capital injections over a finite time horizon. The surplus dynamics obeys a linearly controlled drifted Brownian motion that is reflected at the origin, dividends give rise to time-dependent instantaneous marginal profits, whereas capital injections are subject to time-dependent instantaneous marginal costs. The aim is to maximize the sum of a liquidation value at terminal time and of the total expected profits from dividends, net of the total expected costs for capital injections. Inspired by the study of El Karoui and Karatzas \cite{EKK1989} on reflected follower problems, we relate the optimal dividend problem with capital injections to an optimal stopping problem for a drifted Brownian motion that is absorbed at the origin. We show that whenever the optimal stopping rule is triggered by a time-dependent boundary, the value function of the optimal stopping problem gives the derivative of the value function of the optimal dividend problem. Moreover, the optimal dividend strategy is also triggered by the moving boundary of the associated stopping problem. The properties of this boundary are then investigated in a case study in which instantaneous marginal profits and costs from dividends and capital injections are constants discounted at a constant rate.
\end{abstract}

\maketitle

\smallskip


\smallskip

{\textbf{MSC2010 subject classification}}: 93E20, 60G40, 62P05, 91G10, 60J65

\section{Introduction}
\label{sec:intro}

The literature on optimal dividend problems started in 1957 with the work of de Finetti \cite{deFinetti}, where, for the first time, it was proposed to measure an insurance portfolio by the discounted value of its future dividends' payments. Since then, the literature in Mathematics and Actuarial Mathematics experienced many scientific contributions on the optimal dividend problem, which has been typically modeled as a stochastic control problem subject to different specifications of the control processes and of the surplus dynamics (see, among many others, the early work by Jeanblanc-Piqu\'e and Shiryaev \cite{JS}, the more recent works by Akyildirim et al.\ \cite{Soneretal}, De Angelis and Ekstr\"om \cite{DeAngelis} and Jiang and Pistorius \cite{Pistorius}, the review by Avanzi \cite{Avanzi}, and the book by Schmidli \cite{schmidli2008}).

Starting from the observation that ruin occurs almost surely when the fund's manager pays dividends by following the optimal strategy of de Finetti's problem, Dickson and Waters proposed in \cite{DiksonWaters} several modifications to the original formulation of the optimal dividend problem. In particular, in \cite{DiksonWaters} a model has been suggested in which the shareholders are obliged to inject capital in order to avoid bankruptcy. This is the so-called \emph{optimal dividend problem with capital injections}.

The literature on the optimal dividend problem with capital injections is not as rich as that on the classical de Finetti's problem. In Kulenko and Schmidli \cite{KulenkoSchmidli}, the authors study an optimal dividend problem with capital injections in which the surplus process is reflected at the origin, and on $(0,\infty)$ evolves according to a classical Cram\'er-Lundberg risk model. In Schmidli \cite{Schmidli2016}, an optimal dividend problem with capital injections and taxes in a diffusive setting is formulated and solved. In Lokka and Zervos \cite{LOKKA2008954} the shareholders can choose the capital injections' policy and, in absence of any interventions, the surplus process follows a Brownian motion with drift. Other works in which the surplus process evolves as a general one-dimensional diffusion are the ones by Ferrari \cite{Ferrari17}, Zhu and Yang \cite{Yang}, and Shreve et al.\ \cite{Shreveetal}. Optimal dividends and capital injections in a jump-diffusion setting are determined by Avanzi et al.\ in \cite{Avanzi2007}. In all those papers the optimal dividend problem with capital injections is formulated as a singular stochastic control problem over an infinite time horizon. Given the stationarity of the setting, in those works it is shown that (apart a possible initial lump sum payment) it is usually optimal to pay just enough dividends in order to keep the surplus process in the interval $[0,b]$, for some constant $b>0$ endogenously determined.

In this paper we propose and solve, for the first time in the literature, an optimal dividend problem with capital injections over a finite time horizon $T\in(0,\infty)$. This horizon might be seen as a pre-specified future date at which the fund is liquidated. 

As is common in the literature (see \cite{Soneretal}, \cite{DeAngelis}, and \cite{LOKKA2008954}, among many others), in absence of any interventions, the surplus process evolves as a Brownian motion with drift $\mu$ and volatility $\sigma$. This dynamics for the fund's value can be obtained as a suitable (weak) limit of a classical dynamics \`a la Cram\'er-Lundberg (see Appendix D.3 in \cite{schmidli2008} for details). We also assume that, after time-dependent transaction costs/taxes have been paid, shareholders receive a time-dependent instantaneous net proportion of leakages $f$ from the surplus. Moreover, shareholders are \emph{forced} to inject capital whenever the surplus attempts to become negative, and injecting capital they incur a time-dependent marginal administration cost $m$. Finally, a surplus-dependent liquidation reward $g$ is obtained at liquidation time $T$. Notice, that, under suitable requirements on $f$, $m$ and $g$ (see Remark \ref{remark on dividends 1}), injecting capital at the origin turns out to be optimal within the class of dividends/capital injections that keep the surplus nonnegative for any time with probability one (see also \cite{KulenkoSchmidli}, \cite{ScheerSchmidli2011}, and \cite{Schmidli2016}).

Within this setting, the fund's manager takes the point of view of the shareholders and thus aims at solving 
\begin{equation}
\label{introduction control problem}
V(t,x) := \sup_{D} \mathbb{E}\bigg[\int_0^{T-t} f(t+s) ~dD_s - \int_0^{T-t} m(t+s) ~dI^{D}_s +  g(T,X^{D}_{T-t}(x)) \bigg],
\end{equation} 
for any initial time $t \in [0,T]$ and any initial value of the fund $x \in \mathbb{R}_+$. In \eqref{introduction control problem} the fund's value evolves as
$$X^{D}_s(x)= x + \mu s + \sigma W_s - D_s + I^{D}_s,\quad s \geq 0,$$
and the optimization is performed over a suitable class of nondecreasing processes $D$. In fact, the quantity $D_s$ represents the cumulative amount of dividends paid to shareholders up to time $s$, whereas $I^{D}_s$ is the cumulative amount of capital injected by the shareholders up to time $s$. We take $I^{D}$ as the minimal nondecreasing process which ensures that $X^{D}$ stays nonnegative, and it is flat off $\{t\geq 0:\, X_t^D=0\}$.

If we attempt to tackle problem \eqref{introduction control problem} via a dynamic programming approach, we will find that the dynamic programming equation for $V$ takes the form of a parabolic partial differential equation (PDE) with gradient constraint (i.e.\ a variational inequality), and with a Neumann boundary condition at $x=0$ (the latter is due to the fact that the state process $X$ is reflected at the origin through the capital injections process). Proving that a solution to this PDE problem has enough regularity to characterize an optimal control is far from being trivial. 

Starting from the observation that the optimal dividend problem with capital injections \eqref{introduction control problem} is actually a reflected follower problem (see, e.g., Baldursson \cite{Baldursson}, El Karoui and Karatzas \cite{Karoui1988}, and Karatzas and Shreve \cite{karatzasshrevereflected}) with costly reflection at the origin, and inspired by the results of El Karoui and Karatzas in \cite{EKK1989}, here we solve \eqref{introduction control problem} without relying on PDE methods, but relating \eqref{introduction control problem} to a (still complex but) more tractable optimization problem; i.e., to an optimal stopping problem with absorption at the origin and with value function $u$ (cf.\ \eqref{Stopping problem} below). In this auxiliary optimal stopping problem, the functions $f$, $m$, and $g_x$ give the payoff of immediate stopping, the payoff from absorption at the origin, and the final reward, respectively.

Then, if the optimal stopping time for that problem is given in terms of a continuous and strictly positive time-dependent boundary $b(\,\cdot\,)$ (cf.\ the structural Assumption \ref{Assumption on regions and stopping time} below), one has that $V_x=u$, and the optimal dividends' payments strategy $D^{\star}$ is triggered by $b$ (see Theorem \ref{main theorem} below). In fact, if the optimization starts at time $t \in [0,T]$, the couple $(D^{\star}, I^{D^{\star}})$ keeps at any instant in time $s\in [0,T-t]$ the optimally controlled fund's value $X^{D^{\star}}_s$ nonnegative and below the time-dependent critical level $b(s+t)$.

This result is obtained via an almost exclusively probabilistic study in which we suitably integrate in the space variable two different representations of the value function $u$ of the auxiliary optimal stopping problem. It is worth noticing that although we borrow arguments from the study in \cite{EKK1989} on the connection between reflected follower problems and questions of optimal stopping (see also \cite{karatzasshrevereflected}), differently to \cite{EKK1989}, in our performance criterion \eqref{introduction control problem} we also have a cost of reflection and this requires a careful and not immediate adaptation of the ideas and results of \cite{EKK1989}.

We then show that the structural Assumption \ref{Assumption on regions and stopping time}, needed to establish the relation between \eqref{introduction control problem} and the optimal stopping problem, does indeed hold in a canonical formulation of the optimal dividend problem with capital injections in which marginal benefits and costs are constants discounted at a constant rate, and the liquidation value at time $T$ is proportional to the terminal value of the fund. In particular, we show that the optimal dividend strategy is given in terms of an optimal boundary $b$ that is decreasing, continuous, bounded, and null at terminal time. To the best of our knowledge, also this result appears here for the first time.

The rest of the paper is organized as follows. In Section \ref{sec:problem} we set up the problem, and in Section \ref{sec:mainresult} we state the connection between \eqref{introduction control problem} and the optimal stopping problem with absorption. Its proof is then performed in Section \ref{sec:proofmain}. In Section \ref{sec:casestudy} we consider the case study with (discounted) constant marginal benefits and costs, whereas in the Appendices we collect the proofs of some results needed in the paper.

%%%%%%%%%%%%%%%%%%%%%%%%%%%%%%%%%%%%%%%%%%%%%%%%%%%%%%%%%%%%%%%%%%%%%%%%%%%%%%%%%%%%

\section{Problem Formulation}
\label{sec:problem}

In this section we introduce the optimal dividend problem that is the object of our study. 
Let $\left(\Omega, \mathcal{F}, \mathbb{P}\right)$ be a complete probability space with a filtration $\mathbb{F}:=(\mathcal{F}_s)_{s \geq 0}$ which satisfies the usual conditions.
We assume that the fund's value is described by the one-dimensional process 
\begin{equation}
\label{state space variable}
X^{D}_s(x)= x + \mu s + \sigma W_s - D_s + I^{D}_s,\quad  s \geq 0, 
\end{equation}
where $x\geq 0$ is the initial value of the fund, $\mu \in \mathbb{R}$, $\sigma>0$, and $W$ is an $\mathbb{F}$-standard Brownian motion. For any $s \geq 0$, $D_s$ represents the cumulative amount of dividends paid to shareholders up to time $s$, whereas $I^{D}_s$ is the cumulative amount of capital injected by the shareholders up to time $s$ in order to avoid bankruptcy of the fund. 

\begin{comment}
\begin{remark}
In absence of any dividends' payment and capital injections the fund's value evolves as a Brownian motion with drift $\mu$ and volatility $\sigma$. Such a dynamics is typical in the literature on the optimal dividend problem (see Akyildirim et al.\ \cite{Soneretal}, De Angelis and Ekstr\"om \cite{DeAngelis} and Lokka and Zervos \cite{LOKKA2008954}, among many others), and it can be obtained as a diffusion approximation of a classical risk model \`a la Cram\'er-Lundberg (see Appendix D.3 in Schmidli \cite{schmidli2008} for details).
\end{remark}
\end{comment}

Define the (nonempty) set 
\begin{align*}
\mathcal{A}=\bigg\{ &\nu: \Omega \times \mathbb{R}_+ \to \mathbb{R}_+, \mathbb{F}-\text{adapted s.t. $s \mapsto \nu_s(\omega)$ is a.s.} \\
&\text{nondecreasing and left-continuous, and $\nu_0=0$ a.s.} \bigg\}. 
\end{align*}
For fixed $x \geq 0$, we assume that the fund's manager can pick a dividends' distribution strategy among the processes $D \in \mathcal{A}$ and such that a.s.
\begin{equation}\label{D not jumps under the origin}
D_{s+}-D_s \leq X_s^D(x) \quad \text{ for all } s\geq 0;
\end{equation}
that is, bankruptcy cannot be obtained with a single lump sum dividend's payment.
For any such dividend policy $D$, the capital injections process $I^D$ is given as the minimal cumulative amount of capital needed to ensure that $X^D(x)$ stays nonnegative, and which is flat off $\{t\geq 0: X_t^D(x)=0\}$. In particular, for $x \geq 0$, we take the couple $(X^D(x),I^D)$ as the unique solution to the (discontinuous) Skorokhod reflection problem (see, e.g., Chaleyat-Maurel et al.\ \cite{ElKarouiChaleyat} and Ma \cite{Ma93}):
\begin{equation} \label{Find X and I}
\text{ Find } (X^D(x),I^D) \text{ s.t. } 
\begin{cases}
\displaystyle{I^D \in ~\mathcal{A}, \quad X^{D}_s(x) = x + \mu s + \sigma W_s -D_s+I^{D}_s, ~~ s \geq 0, } \vspace{0.1cm}\\
\displaystyle{X^{D}_s(x) \geq 0 \quad \text{a.s. for any} ~ s\geq 0, } \vspace{0.1cm} \\
\displaystyle{\int_0^{\infty} X^D_s(x) d(I^{D}_s)^c =0 \quad \text{a.s.},} \vspace{0.1cm} \\
\displaystyle{ \Delta  I^{D}_s:=I^{D}_{s+} -I^{D}_s=2X^D_{s+}(x) \quad \forall s \in \{s\geq 0: \Delta I^{D}_s >0\}.}
\end{cases}
\end{equation}
Here, $(I^{D})^c$ denotes the continuous part of $I^{D}$. Notice that, given \eqref{D not jumps under the origin}, the process 
\begin{equation}\label{def of I without dividends}
I_t^D := 0 \vee \sup_{0 \leq s \leq t} \left( D_s-(x+ \mu s+ \sigma W_s)\right), \quad t\geq0, \quad I_0^D=0,
\end{equation}
uniquely solves \eqref{Find X and I} and $t \mapsto I_t^D$ is continuous (see, e.g., Propositions 2 and 3 in \cite{ElKarouiChaleyat}, or Theorem 3.1 and Corollary 3.2 in \cite{karatzasshrevereflected}). As a consequence, the last condition in \eqref{Find X and I} is not binding, since $\Delta I_t^D=0$ a.s.\ for all $t \geq 0$.

Given a time horizon $T \in (0,\infty)$ representing, e.g., a finite liquidation time, the fund's manager takes the point of view of the shareholders, and is faced with the problem of choosing a dividends' distribution strategy $D$ maximizing the performance criterion
\begin{equation} \label{functional of the optimal control problem}
\mathcal{J}(D;t,x)=\mathbb{E}\left[\int_0^{T-t} f(t+s) ~dD_s - \int_0^{T-t} m(t+s) ~dI^{D}_s +  g(T,X^{D}_{T-t}(x)) \right],
\end{equation}
for $(t,x)\in [0,T]\times \mathbb{R}_+$ given and fixed. That is, the fund's manager aims at solving
\begin{equation}\label{value function of ocp}
V(t,x):=\sup_{D\in \mathcal{D}(t,x)} \mathcal{J}(D;t,x), \quad (t,x) \in [0,T] \times \mathbb{R}_+.
\end{equation}
Here, for any $(t,x)\in [0,T]\times \mathbb{R}_+$, $\mathcal{D}(t,x)$ denotes the class of dividend payments belonging to $\mathcal{A}$ and satisfying \eqref{D not jumps under the origin}, when the surplus process $X^D$ starts from level $x$ and the optimization runs up to time $T-t$. In the following, any $D \in \mathcal{D}(t,x)$ will be called \emph{admissible} for $(t,x) \in [0,T] \times \mathbb{R}_+$.

In the reward functional \eqref{functional of the optimal control problem} the term $\mathbb{E}[ \int_0^{T-t} f(t+s)~dD_s ]$ is the total expected cash-flow from dividends. The function $f$ might be seen as a time-dependent instantaneous net proportion of leakages from the surplus received by the shareholders after time-dependent transaction costs/taxes have been paid.
The term $\mathbb{E}[ \int_0^{T-t} m(t+s)~dI^D_s ]$ gives the total expected costs of capital injections, and $m$ is a time-dependent marginal administration cost for capital injections. Finally, $\mathbb{E}\left[ g(T,X^D_{T-t}(x)) \right]$ is a liquidation value.

The functions $f$, $m$, and $g$ satisfy the following conditions.
\begin{assumption}
\label{basic assumptions}
$f:[0,T] \to \mathbb{R}_+, m:[0,T] \to \mathbb{R}_+, g:[0,T] \times \mathbb{R}_+ \to \mathbb{R}_+$ are continuous, $f$ and $m$ are continuously differentiable with respect to $t$, and $g$ is continuously differentiable with respect to $x$. Moreover,
\begin{itemize}
\item[(i)] $g_x(T,x)\geq f(T) \quad \mbox{for any}\,\,x \in (0,\infty)$,
\item[(ii)] $m(t) > f(t) \quad \mbox{for any}\,\,t \in [0,T]$.
\end{itemize}
\end{assumption}

\begin{remark}
Requirement $(i)$ ensures that the marginal liquidation value is at least as high as the marginal profits from dividends. This will ensure that the value function of the optimal stopping problem considered below is not discontinuous at terminal time.

Condition $(ii)$ means that the marginal costs for capital injections are bigger than the marginal profits from dividends. Notice that in the case in which $m<f$ the value function might be infinite, as it shown in the next example.
Take $f(s)=\eta,~m(s)=\kappa$ for all $s \in [0,T]$, and $\eta > \kappa$. For arbitrary $\beta > 0$ consider the admissible strategy $\widehat{D}_s:=\beta s$, and notice that $\widehat{I}_s^D=\sup_{0 \leq u \leq s} (-x-\mu u - \sigma B_u +\beta u)\vee 0$. Then $\widehat{I}_s^D\leq \beta s + Y_s$, with $Y_s:=\sup_{0 \leq u \leq s} (-x-\mu u - \sigma B_u)\vee 0$, and using that $g\geq0$ we obtain for the sub-optimal strategy $\widehat{D}$
\begin{align}
V(t,x)
&\geq \beta \eta (T-t)- \beta \kappa (T-t)-\kappa \mathbb{E}[Y_{T-t}] \nonumber \\
&=\beta (T-t) (\eta-\kappa)-\kappa \mathbb{E}[Y_{T-t}]. \nonumber
\end{align}
However, the latter expression can be made arbitrarily large by increasing $\beta$ if $\eta>\kappa$. 

On the other hand, by taking $m(t)=f(t)=e^{-rt}$, is has been recently shown in Ferrari \cite{Ferrari17} for a problem with $T=+\infty$ (see Theorem 3.8 therein) that an optimal control may not exist, but only an $\varepsilon$-optimal control does exist.

In order to avoid pathological situations as the ones described above, here we assume Assumption \ref{basic assumptions}-(ii).
\end{remark}

\begin{remark}
\label{rem:discountedform}
Notice that our formulation is general enough to accommodate also a problem in which profits and costs are discounted at a deterministic time-dependent discount rate $(r_s)_{s\geq 0}$. Indeed, if we consider the optimal dividend problem with capital injections
\begin{eqnarray}
\label{rem:eqrate}
\widehat{V}(t,x)
&\hspace{-0.25cm} := \hspace{-0.25cm} &\sup_{D \in \mathcal{D}(t,x)} \mathbb{E}\bigg[ \int_0^{T-t} e^{-\int_{t}^{t+s}r_{\alpha} d\alpha}\, \widehat{f}(t+s)~dD_s - \int_0^{T-t} e^{-\int_{t}^{t+s}r_{\alpha} d\alpha}\, \widehat{m}(t+s)~dI^{D}_s \nonumber \\
&& \hspace{2cm} +\, e^{-\int_t^T r_{\alpha} d\alpha}\widehat{g}(T,X_{T-t}^D(x)) \bigg], \nonumber 
\end{eqnarray}
then, for any $(t,x) \in [0,T] \times \mathbb{R}_+$ we can set 
$$f(t):= e^{-\int_{0}^{t}r_{\alpha} d\alpha}\, \widehat{f}(t), \quad m(t):= e^{-\int_{0}^{t}r_{\alpha} d\alpha}\, \widehat{m}(t), \quad g(t,x):=e^{-\int_0^t r_{\alpha} d\alpha}\widehat{g}(t,x),$$
and $V(t,x):=e^{-\int_{0}^{t}r_{\alpha} d\alpha} \,\widehat{V}(t,x)$ is of the form \eqref{value function of ocp}.

In Section \ref{sec:casestudy} we will consider a problem with constant marginal profits and costs discounted at a constant rate $r>0$ (see \eqref{ex:1}, \eqref{control problem in the example} and \eqref{ex:data} in Section \ref{sec:casestudy}).
\end{remark}

\begin{remark}
\label{remark on dividends 1}
Notice that in our model shareholders are \emph{forced} to inject capital whenever the surplus process attempts to become negative; that is, the capital injection process is not a control variable of their, and shareholders do not choose when and how to invest in the company. 

Injecting capital at the origin, under the condition that bankruptcy is not allowed, can be shown to be optimal in the canonical formulation of the optimal dividend problem of Section \ref{sec:casestudy} in which marginal costs and profits are constants discounted at a constant interest rate. Indeed, in such a case, due to discounting, shareholders will inject capital as late as possible in order to minimize the total costs of capital injections. See also Kulenko and Schmidli \cite{KulenkoSchmidli} and Schmidli \cite{Schmidli2016} for a similar result in stationary problems. More in general, the policy ``inject capital at the origin'' is optimal when $m$ is decreasing and $\min_{t \in [0,T]}m(t)>g_x(T,x)$ for all $x\in \mathbb{R}_+$. Under these conditions, shareholders postpone injection of capital, and inject only as much capital as necessary since any additional capital injection cannot be compensated by the reward at terminal time.
\end{remark}

The dynamic programming equation for $V$ takes the form of a parabolic partial differential equation (PDE) with gradient constraint, and with a Neumann boundary condition at $x=0$ (the latter is due to the fact that the state process $X$ is reflected at the origin through the capital injections process). Indeed, it reads
$$\max\Big\{\partial_t V + \frac{1}{2}\sigma^2 \partial_{xx}V + \mu \partial_x V, f - \partial_{x}V \Big\} =0, \quad \mbox{on}\quad [0,T) \times (0,\infty),$$
with boundary conditions $\partial_{x}V(0,t)=m(t)$ for all $t\in [0,T]$, and $V(T,x) = g(T,x)$ for any $x \in (0,\infty)$.
Proving that such a PDE problem admits a solution that has enough regularity to characterize an optimal control is far from being trivial. 

In order to solve the optimal dividend problem \eqref{value function of ocp} we then follow a different approach, and we relate \eqref{value function of ocp} to an optimal stopping problem with absorbing condition at $x=0$. This is obtained by borrowing arguments from the study of El Karoui and Karatzas in \cite{EKK1989} on the connection between reflected follower problems and questions of optimal stopping (see also Baldursson \cite{Baldursson} and Karatzas and Shreve \cite{karatzasshrevereflected}). However, differently to \cite{EKK1989}, in our performance criterion \eqref{functional of the optimal control problem} we also have a cost of reflection which requires a careful and not immediate adaptation of the ideas and results of \cite{EKK1989}. 

In particular, introducing a problem of optimal stopping with absorption at the origin, we show that a proper integration of the value function of the latter leads to the value function of the optimal control problem \eqref{value function of ocp}. This result is stated in the next section, and then proved in Section \ref{sec:proofmain}.

%%%%%%%%%%%%%%%%%%%%%%%%%%%%%%%%%%%%%%%%%%%%%%%%%%%%%%%%%%%%%%%%%%%%%%%%%%%%%%%%%%%%

\section{The Main Result}
\label{sec:mainresult}

Let $S(x):=\inf\{s\geq 0: x + \mu s + \sigma W_s = 0\}$, $x\geq 0$, and for any $s\geq 0$, introduce the absorbed drifted Brownian motion 
\begin{equation} 
\label{absorbed process}
A_s(x):=
\begin{cases}
x + \mu s +\sigma W_s, &s<S(x),\\
\Delta, &s \geq S(x),
\end{cases}
\end{equation}
where $\Delta$ is a cemetery state isolated from $\mathbb{R}_+$ (i.e.\ $\Delta < 0$). 

Introducing the convention $g_x(T,\Delta):=0$, for $(t,x)\in [0,T]\times \mathbb{R}_+$, consider the optimal stopping problem 
\begin{equation}
\begin{aligned}
\label{Stopping problem}
u(t,x)
&:= \sup_{\tau \in [0,T-t]} \mathbb{E}\Big[f(t+\tau) \mathbbm{1}_{\left\{\tau < (T-t) \wedge S(x) \right\}}+m(t+S(x))  \mathbbm{1}_{\left\{\tau \geq S(x)\right\}} \\
&\hspace{1.5cm} +  g_x\big(T,x+\mu (T-t)+\sigma W_{T-t}\big) \mathbbm{1}_{\left\{\tau=T-t < S(x)\right\}} \Big]  \\
&=\sup_{\tau \in \Lambda(T-t)} \mathbb{E}\Big[f(t+\tau) \mathbbm{1}_{\left\{A_{\tau}(x) >0 \right\}} \mathbbm{1}_{\left\{\tau < T-t \right\}} +m(t+S(x))  \mathbbm{1}_{\left\{A_{\tau}(x)\leq 0\right\}}  \\
&\hspace{1.5cm} +  g_x\big(T,A_{T-t}(x)\big) \mathbbm{1}_{\left\{\tau=T-t\right\}}  \Big],  
%&=: \sup_{\tau \in \Lambda(T-t)} \mathbb{E}\Big[H(t+(\tau \wedge S(x)),A_{\tau}(x))\Big],
\end{aligned}
\end{equation} 
where $\Lambda(T-t)$ denotes the set of all $\mathbb{F}$-stopping times with values in $[0,T-t]$ a.s. 
%and $H(t,x):=f(t) \mathbbm{1}_{\left\{x >0 \right\}} \mathbbm{1}_{\left\{t < T \right\}} +m(t)\mathbbm{1}_{\left\{x \leq 0 \right\}}+g_x(T,x)\mathbbm{1}_{\left\{t=T \right\}}\mathbbm{1}_{\left\{x>0 \right\}}$. 
Problem \eqref{Stopping problem} is an optimal stopping problem for the absorbed process $A$.  

To establish the relation between \eqref{value function of ocp} and \eqref{Stopping problem} we need the following \emph{structural assumption}, which will be standing in this section and in Section \ref{sec:proofmain}. Its validity has to be verified on a case by case basis. In particular, it holds in the optimal dividend problem considered in Section \ref{sec:casestudy}. 

\begin{assumption} 
\label{Assumption on regions and stopping time}
Assume that the continuation region of the stopping problem \eqref{Stopping problem} is given by
\begin{equation}
\mathcal{C}:=\left\{(t,x)\in [0,T)\times (0,\infty): u(t,x)>f(t) \right\}=\left\{(t,x)\in [0,T) \times (0,\infty) : x<b(t)\right\}, \label{continuation region}
\end{equation}
and that its stopping region by
\begin{align}
\mathcal{S}
&:=\left\{(t,x)\in [0,T)\times (0,\infty): u(t,x)\leq f(t) \right\} \cup \big(\{T\} \times (0,\infty)\big) \nonumber \\
&=\left\{(t,x)\in [0,T)\times (0,\infty) : x\geq b(t)\right\}\cup \big(\{T\} \times (0,\infty)\big), \label{stopping region}
\end{align}
for a continuous function $b: [0,T) \to (0,\infty)$. We refer to the function $b$ as to the optimal stopping boundary of problem \eqref{Stopping problem}. Further, assume that the stopping time

\begin{equation}
\tau^{\star}(t,x):=\inf\{s \in [0,T-t): A_s(x) \geq b(t+s)\}\wedge (T-t) \label{optimal stopping time}
\end{equation}
(with the usual convention $\inf \emptyset = + \infty$) is optimal; that is,
\begin{align}
u(t,x)
&=  \mathbb{E}\Big[f(t+\tau^{\star}(t,x)) \mathbbm{1}_{\left\{\tau^{\star}(t,x) < (T-t) \wedge S(x) \right\}}+m(t+S(x))  \mathbbm{1}_{\left\{\tau^{\star}(t,x) \geq S(x)\right\}} \nonumber \\
&+  g_x(T,x+\mu (T-t)+\sigma W_{T-t}) \mathbbm{1}_{\left\{\tau^{\star}(t,x)=T-t < S(x)\right\}} \Big]. 
\label{OST with optimal stopping time}
\end{align}
\end{assumption}

For any $(t,x) \in [0,T] \times \mathbb{R}_+$, and with $b$ the optimal stopping boundary of problem \eqref{Stopping problem} (cf.\ Assumption \ref{Assumption on regions and stopping time}), we define the processes $I^{\star}(t,x)$ and $D^{\star}(t,x)$ through the system 
\begin{align}
\begin{cases}
\displaystyle D^{\star}_s(t,x):=\max\left\{0,\max_{0 \leq \theta \leq s}\big(x+\mu \theta + \sigma W_{\theta}+I^{\star}_{\theta}(t,x)-b(t+\theta)\big)\right\},  \vspace{0.17cm}
\\
\displaystyle I^{\star}_s(t,x):=\max\left\{0,\max_{0 \leq \theta \leq s}\big(-x-\mu \theta - \sigma W_{\theta}+D^{\star}_{\theta}(t,x)\big)\right\}, \label{system for I and D}
\end{cases}
\end{align}
for any $s \in [0,T-t]$, and with initial values $D^{\star}_0(t,x)=I^{\star}_0(t,x)=0$ a.s. The existence and uniqueness of the solution to system $\eqref{system for I and D}$ can be proved by an application of Tarski's fixed point theorem following arguments as those employed by Karatzas in the proof of Proposition 7 in Section 8 of \cite{karatzas83}. It can be easily shown from \eqref{system for I and D} and the positivity of $b$ that $D^{\star}$ satisfies \eqref{D not jumps under the origin}, and, consequently, that $I^{\star}$ has continuous paths. The latter property of $I^{\star}$ implies that $t \mapsto D_t^{\star}$ is continuous apart for a possible initial jump at time zero of amplitude $(x-b(t))^+$. We can now state the following result.

\begin{theorem}
\label{main theorem}
Let Assumption \ref{Assumption on regions and stopping time} hold. Then, the process $D^{\star}$ defined through \eqref{system for I and D} provides the optimal dividends' distribution policy, and the value function $V$ of \eqref{value function of ocp} is such that
\begin{equation}
\label{eq:main}
V(t,x)=V(t,b(t))-\int_x^{b(t)} u(t,y)~dy, \quad (t,x) \in [0,T] \times \mathbb{R}_+.
\end{equation}
Assume further that $\lim_{t \uparrow T}b(t)=:b(T) < \infty$. Then
\begin{align}
\label{eq:V such that one can compute it}
V(t,b(t)) &=-\mu \int_0^{T-t}f^{\prime}(t+s)s ~ds + \int_0^{T-t}f^{\prime}(t+s)b(t+s) ~ds  \nonumber \\
&+g(T,b(T))+f(T)\mu(T-t)+ f(t)b(t) -f(T)b(T).
\end{align}
\end{theorem}

Consistently with the result of El Karoui and Karatzas in \cite{EKK1989} (see also Karatzas and Shreve \cite{karatzasshrevereflected}), we find that also in our problem with costly reflection at the origin the value of an optimal stopping problem (namely, problem \eqref{Stopping problem}) gives the marginal value of the value function \eqref{value function of ocp}. The optimal stopping boundary $b$ thus triggers the timing at which it is optimal to pay an additional unit of dividends. Moreover, once the optimal stopping value function $u$ and its corresponding free boundary $b$ are known, \eqref{eq:main} and \eqref{eq:V such that one can compute it} provide a complete characterization of the optimal dividend problem's value function $V$. Notice that the condition $b(T)<\infty$ is satisfied in the case study of Section \ref{sec:casestudy}, where we actually prove that $b(T)=0$. The proof of Theorem \ref{main theorem} is quite lengthy and technical, and it is relegated to Section \ref{sec:proofmain}.

%%%%%%%%%%%%%%%%%%%%%%%%%%%%%%%%%%%%%%%%%%%%%%%%%%%%%%%%%%%%%%%%%%%%%%%%%%%%%%%%%%%%

\section{On the Proof of Theorem \ref{main theorem}}
\label{sec:proofmain}

This section is entirely devoted to the proof of Theorem \ref{main theorem}. This is done through a series of intermediate results which are proved by employing mostly probabilistic arguments. Assumption \ref{Assumption on regions and stopping time} will be standing throughout this section.

\subsection{On a Representation of the Optimal Stopping Value Function}

Here we derive an alternative representation for the value function of the optimal stopping problem \eqref{Stopping problem}, by borrowing ideas from El Karoui and Karatzas \cite{EKK1989}, Section 3. In the following we set $g_x(T,\Delta)=0$. 

The idea that we adopt here is to rewrite the optimal stopping problem \eqref{Stopping problem} in terms of the function $b$ of Assumption \ref{Assumption on regions and stopping time}.
To accomplish that, for given $(t,x) \in [0,T] \times \mathbb{R}_+$, define the payoff associated to the admissible stopping rule ``never stop" as
\begin{equation}
G(t,x):=\mathbb{E}\left[m(t+S(x)) \mathbbm{1}_{\{S(x) \leq T-t\}} + g_x(T,A_{T-t}(x)) \right], \label{definition of G}
\end{equation}
where we have used that $g_x(T,A_{T-t}(x))\mathbbm{1}_{\{T-t < S(x)\}}=g_x(T,A_{T-t}(x))$ because of \eqref{absorbed process} and the fact that $g_x(T,\Delta)=0$.
%\begin{equation}
%G(t,x):=\mathbb{E}\left[m(t+S(x)) \mathbbm{1}_{\{A_{T-t} =\Delta \}} + g_x(T,A_{T-t}(x)) \right]. \label{definition of G}
%\end{equation}

Also, introduce the function $\tilde{g}: [0,T] \times \mathbb{R}_+ \times \mathbb{R}_+ \to \mathbb{R}$ (depending parametrically on $t$) as 
\begin{equation} 
\label{definition of g tilde}
\tilde{g}(\alpha,q,y;t):=
\begin{cases}
g_x (T,y), & \alpha < q, \\
m(t+q), & \alpha \geq q,
\end{cases}
\end{equation}
and notice that $v:=u-G$ admits the representation
\begin{equation}
v(t,x)=\sup_{\tau \in \Lambda(T-t)} \mathbb{E} \left[ \left( f(t+\tau)-\tilde{g}(T-t, S(x), A_{T-t}(x);t) \right) \mathbbm{1}_{ \left\{\tau < S(x) \wedge T-t\}\right\}}\right]. \label{first representation of v}
\end{equation}
%Also, introduce the function $\tilde{g}: [0,T] \times [0,\infty] \times \mathbb{R}_+ \to \mathbb{R}$ (depending parametrically on $(t,x)$) as 
%\begin{equation} \label{definition of g tilde}
%\tilde{g}(\alpha,q,y;t,x):=\begin{cases}
 %g_x (T,y), & \alpha < q \\
%m(t+S(x)), & \alpha \geq q,
%\end{cases}
%\end{equation}
%and notice that $v:=u-G$, due to the fact that $g_x(T,\Delta)=0$ and $A_{\tau}>0 \Leftrightarrow \tau < S(x)$ by \eqref{absorbed process}, admits the representation
%\begin{equation}
%v(t,x)=\sup_{\tau \in \Lambda(T-t)} \mathbb{E} \left[ \left( f(t+\tau)-\tilde{g}(T-t, S(x), A_{T-t}(x);t,x) \right) \mathbbm{1}_{ \left\{\tau < S(x) \wedge T-t\}\right\}}\right]. \label{first representation of v}
%\end{equation}

Clearly, the stopping time $\tau^{\star}$ defined by \eqref{optimal stopping time} is also optimal for $v$ since $G$ is independent of $\tau \in \Lambda(T-t)$. Therefore, we can expect that $v$ can be expressed in terms of the optimal stopping boundary $b$. Following \cite{EKK1989}, we obtain such a representation for $v$ by means of the theory of dual previsible projections (``balay\'{e}e pr\'{e}visible"), as it is shown in the following.
From now on, $(t,x) \in [0,T] \times \mathbb{R}_{+}$ will be given and fixed.

We define the process $(C_\alpha)_{\alpha \in [0,T]}$ such that for any $\alpha \in [0,T-t]$
\begin{align}
C_\alpha(t,x)
&:=-\int_0^{\alpha \wedge S(x) \wedge T-t}f^{\prime}(t+\theta) d\theta \nonumber \\
&+\big[f(T \wedge (t+S(x))) -\tilde{g}(T-t, S(x), A_{T-t}(x);t) \big]\mathbbm{1}_{ \{ 0 < T-t \wedge S(x) \leq \alpha\}}, \label{definition of C}
\end{align}
as well as the stopping time 
\begin{equation}
\sigma_\alpha(t,x):=\inf\left\{\theta \in [\alpha,T-t): A_{\theta}(x)\geq b(t+\theta)\right\} \wedge (T-t), \label{definition of stopping time sigma}
\end{equation}
with the convention $\inf \emptyset = + \infty$.
The process $C_{\cdot}(t,x)$ is absolutely continuous on $[0,T-t) \wedge S(x)$ with a possible jump at $(T-t) \wedge S(x)$, and $\alpha \mapsto \sigma_{\alpha}(t,x)$ is a.s.\ nondecreasing and right-continuous.

Since the stopping time $\sigma_0(t,x)$ is optimal for $u(t,x)$ by Assumption \ref{Assumption on regions and stopping time}, and therefore also for $v(t,x)=(u-G)(t,x)$, by using \eqref{definition of C} we can write from \eqref{first representation of v}
\begin{equation}\label{v in terms of C}
v(t,x)=\mathbb{E}\left[C_{T-t}(t,x)-C_{\sigma_0(t,x)}(t,x)\right]=\mathbb{E}\Big[\widetilde{C}_{T-t}(t,x)\Big],
\end{equation}
where we have introduced 
\begin{equation}
\widetilde{C}_{\alpha}(t,x):=C_{\sigma_\alpha(t,x)}(t,x)-C_{\sigma_0(t,x)}(t,x), \quad \alpha \in [0,T-t]. \label{definition of C tilde}
\end{equation}

The process $\widetilde{C}_{\cdot}(t,x)$ is of bounded variation, since it is the composition of the process of bounded variation $C_{\cdot}(t,x)$ and of the nondecreasing process $\sigma_{\cdot}(t,x)$, but it is not $\mathbb{F}$-adapted. However, being $v$ an excessive function, it is also the potential of an adapted, nondecreasing process $\Theta_{\cdot}(t,x)$ (cf.\ Section IV.4 in the book of Blumenthal and Getoor \cite{BlumenthalGetoor}), which is the dual predictable (or previsible) projection of $\widetilde{C}_{\cdot}(t,x)$ (see, e.g., Theorem 21.1 in Chapter VI of the book by Rogers and Williams \cite{Rogers2000} for further details on the dual predictable projection). In the following we provide the explicit representation of $\Theta_{\cdot}(t,x)$. This is obtained by employing the methodology of El Karoui and Karatzas in \cite{EKK1991}, Section 7.

\begin{theorem}
\label{theorem dual predictable projection}
The dual predictable projection $\Theta(t,x)$ of $\widetilde{C}(t,x)$ exists, is nondecreasing and it is given by
\begin{align}
\Theta_\alpha(t,x)
&=\int_0^{\alpha} -f^{\prime}(t+\theta)\mathbbm{1}_{\{A_{\theta}(x) > b(t+\theta)\}}~d\theta   \nonumber \\
&+  \Big[f(T \wedge (t+S(x)))-\tilde{g}(T-t,S(x),A_{T-t}(x);t)\Big] \mathbbm{1}_{\{ A_{T-t}(x) > b(T) \}} \mathbbm{1}_{\{ 0 < T-t \wedge S(x) \leq \alpha \}} \label{definition of D} \\
&=\int_0^{\alpha \wedge S(x)} -f^{\prime}(t+\theta)\mathbbm{1}_{\{x + \mu \theta + \sigma W_{\theta} > b(t+\theta)\}}~d\theta   \nonumber \\
&+  \Big[f(T \wedge (t+S(x)))-\tilde{g}(T-t,S(x),A_{T-t}(x);t)\Big] \mathbbm{1}_{\{ A_{T-t}(x) > b(T) \}} \mathbbm{1}_{\{ 0 < T-t \wedge S(x) \leq \alpha \}} \nonumber
\end{align}
for any $\alpha \in [0,T-t]$.
\end{theorem}

Theorem \ref{theorem dual predictable projection} can be proved by carefully adapting to our case the techniques presented in Section 7 of \cite{EKK1991} (see also, Section 3 of \cite{EKK1989}). In particular, differently to Section 7 of \cite{EKK1991}, here we deal with an absorbed drifted Brownian motion as a state variable of the optimal stopping problem \eqref{Stopping problem} (instead of a Brownian motion). However, all the arguments and proofs of Section 7 of \cite{EKK1991} carry over also to our setting with random time horizon $(T-t)\wedge S(x)$ (up to which the process $A$ is in fact a drifted Brownian motion) upon using representation \eqref{first representation of v} of $v$ (in which the function $\tilde{g}$ takes care of the random time horizon $(T-t)\wedge S(x)$) together with \eqref{definition of stopping time sigma} and \eqref{definition of C tilde}.

A consequence of Theorem \ref{theorem dual predictable projection} is the next result.
\begin{corollary}
\label{corollary for D}
It holds that
\begin{itemize}
\item[(i)]
$\big[f(T \wedge (t+S(x)))-\tilde{g}(T-t,S(x),A_{T-t}(x);t)\big] \mathbbm{1}_{\{ A_{T-t}(x) > b(T) \}} = 0$ a.s. 
\vspace{0.15cm}

\item[(ii)] $\{t \in [0,T): f^{\prime}(t) \leq 0\} \supseteq \mathcal{S}$;
\end{itemize}
\end{corollary}

\begin{proof}
(i) On the set $\{ A_{T-t}(x) > b(T) \}$ we obtain by the definition of $\tilde{g}$ (see \eqref{definition of g tilde}) that
\begin{align}
\label{equation f must be bigger then g}
f(T \wedge (t+S(x)))-\tilde{g}(T-t,S(x),A_{T-t}(x);t)
&=f(T)-g_x(T,A_{T-t}(x)).
\end{align}
Since $\Theta_{\cdot}(t,x)$ is nondecreasing, the last term in \eqref{equation f must be bigger then g} has to be positive, thus implying $f(T)-g_x(T,A_{T-t}(x)) \geq 0$ on $\{A_{T-t}(x) > b(T)\}$. However, by Assumption \ref{basic assumptions}-(i) one has $f(T)\leq g_x(T,x)$ for all $x \in (0,\infty)$. Hence the claim follows.
\vspace{0.15cm}

(ii) Since $\alpha \mapsto \Theta_{\alpha}(t,x)$ is a.s.\ nondecreasing, it follows from (i) above and \eqref{definition of D} that $f^{\prime}(t+\theta)\mathbbm{1}_{\{A_{\theta}(x) > b(t+\theta)\}} \leq 0$ a.s.\ for a.e. $\theta \in [0,T-t]$. But $f^{\prime}(\cdot)$, $A_{\cdot}(x)$ and $b(t+\cdot)$ are continuous up to $(T-t) \wedge S(x)$, and therefore the latter actually holds a.s.\ for all $\theta \in [0,T-t]$.  Hence, $\{t \in [0,T): f^{\prime}(t) \leq 0\} \supseteq \mathcal{S}$. 
\end{proof}

\begin{remark}
\label{rem:gxfT}
As a byproduct of Corollary \ref{corollary for D}-(i) (see in particular \eqref{equation f must be bigger then g}), Assumption \ref{basic assumptions}-(i), and of the fact that $A_{T-t}(x)$ has a transition probability that is absolutely continuous with respect to the Lebesgue measure on $\mathbb{R}_+$ (cf.\ \eqref{density of the absorbed process}), one has $\big(f(T) - g_x(T,y)\big)\mathbbm{1}_{\{y > b(T)\}}=0$ for $y \geq 0$.
\end{remark}

We can now obtain an alternative representation of the value function $u$ of problem \eqref{Stopping problem}.

\begin{theorem}
For any $(t,x) \in [0,T] \times \mathbb{R}_+$ one has
\begin{align}
u(t,x)
&=\mathbb{E}\bigg[ \int_0^{(T-t)\wedge S(x)} -f^{\prime}(t+\theta)\mathbbm{1}_{\{x + \mu \theta + \sigma W_{\theta} \geq b(t+\theta)\}}~d\theta  \nonumber \\
&+ m(t+S(x))\mathbbm{1}_{\{S(x) \leq T-t\}}+g_x(T,A_{T-t}(x)) \bigg]. \label{second representation of u}
\end{align}
\end{theorem}

\begin{proof}
Since by Theorem \ref{theorem dual predictable projection} $\Theta(t,x)$ is the dual predictable projection of $\widetilde{C}(t,x)$, from \eqref{v in terms of C} we can write for any $(t,x) \in [0,T] \times \mathbb{R}_+$
\begin{equation}
\label{v:prevproj}
v(t,x)=\mathbb{E}\left[ \widetilde{C}_{T-t}(t,x) \right]=\mathbb{E}\left[ \Theta_{T-t}(t,x) \right]. 
\end{equation}

Due to \eqref{definition of D} and Corollary \ref{corollary for D}-(i), \eqref{v:prevproj} gives
\begin{align}
v(t,x)=\mathbb{E}\left[ \int_0^{(T-t)\wedge S(x)} -f^{\prime}(t+\theta)\mathbbm{1}_{\{x + \mu \theta + \sigma W_{\theta} \geq b(t+\theta)\}}~d\theta \right]. 
\label{second representation of v}
\end{align}
Here we have also used that the joint law of $S(x)$ and of the drifted Brownian motion is absolutely continuous with respect to the Lebesgue measure in $\mathbb{R}^2$ (cf.\ \eqref{first density for u}) to replace $\mathbbm{1}_{\{x + \mu \theta + \sigma W_{\theta} > b(t+\theta)\}}$ with $\mathbbm{1}_{\{x + \mu \theta + \sigma W_{\theta} \geq b(t+\theta)\}}$ inside the expectation in \eqref{definition of D}.

However, since by definition $v=u-G$, we obtain from \eqref{second representation of v} and \eqref{definition of G} the alternative representation
\begin{align}
u(t,x)
&=v(t,x)+G(t,x) \nonumber =\mathbb{E}\bigg[ \int_0^{(T-t)\wedge S(x)} -f^{\prime}(t+\theta)\mathbbm{1}_{\{x + \mu \theta + \sigma W_{\theta} \geq b(t+\theta)\}}~d\theta  \nonumber \\
&+ m(t+S(x))\mathbbm{1}_{\{S(x) \leq T-t\}}+g_x(T,A_{T-t}(x)) \bigg]. \nonumber
\end{align}
\end{proof}

\begin{remark}
Notice that representation \eqref{second representation of u} coincides with that one might obtain by an application of It\^{o}'s formula if $u$ were $C^{1,2}([0,T)\times (0,\infty)) \cap C([0,T] \times \mathbb{R}_+)$, and satisfies (as it is customary in optimal stopping problems) the free-boundary problem
\begin{equation}
\label{FBP}
\begin{cases}
\partial_t u + \frac{1}{2}\sigma^2 \partial^2_{xx} u +\mu \partial_x u =0,\quad &0<x<b(t),\,\, t \in [0,T) \\
u=f, \quad &x \geq b(t),\,\, t \in [0,T) \\
u(T,x)=g_x(T,x), \quad &x>0 \\
u(t,0)=m(t), \quad &t \in [0,T].
\end{cases}
\end{equation}
Indeed, in such a case an application of Dynkin's formula gives 
\[\mathbb{E}\left[ u(t+(T-t)\wedge S(x),Z_{(T-t)\wedge S(x)}(x))\right] = u(t,x) + \mathbb{E}\left[\int_0^{(T-t)\wedge S(x)} f^{\prime}(t+\theta) \mathbbm{1}_{\{ Z_{\theta}(x) \geq b(t+\theta) \}}~d\theta\right],
\]
where we have set $Z_s(x):=x + \mu s + \sigma W_s$, $s\geq 0$, to simplify exposition. Hence, using \eqref{FBP} we have from the latter
\begin{align}
&u(t,x)
= \mathbb{E}\bigg[ m(t+S(x)) \mathbbm{1}_{\{ S(x) \leq T-t \}} +g_x(T, x + \mu (T-t) + \sigma W_{T-t}) \mathbbm{1}_{\{ S(x) > T-t \}}  \nonumber \\
&- \int_0^{(T-t)\wedge S(x)} f^{\prime}(t+\theta) \mathbbm{1}_{\{ Z_{\theta}(x) \geq b(t+\theta) \}}~d\theta \bigg] \nonumber  = \mathbb{E}\bigg[ m(t+S(x)) \mathbbm{1}_{\{ S(x) \leq T-t \}} \nonumber \\
&+g_x(T, A_{T-t}(x)) \mathbbm{1}_{\{ S(x) > T-t \}} - \int_0^{(T-t)\wedge S(x)} f^{\prime}(t+\theta) \mathbbm{1}_{\{ Z_{\theta}(x) \geq b(t+\theta) \}}~d\theta \bigg] \nonumber \\
&= \mathbb{E}\bigg[ m(t+S(x)) \mathbbm{1}_{\{S(x) \leq T-t\}}
+g_x(T, A_{T-t}(x)) - \int_0^{(T-t)\wedge S(x)} f^{\prime}(t+\theta) \mathbbm{1}_{\{ Z_{\theta}(x) \geq b(t+\theta) \}}~d\theta \bigg], \nonumber  
\end{align}
where in the last step we have used that $g_x(T,A_{T-t}(x))\mathbbm{1}_{\{S(x)>T-t\}}=g_x(T,A_{T-t}(x))$ because of \eqref{absorbed process} and the fact that $g_x(T,\Delta)=0$.
\end{remark}

\begin{remark}
\label{rem:inteq}
Notice that the representation \eqref{second representation of u} immediately gives an integral equation for the optimal stopping boundary $b$. Indeed, since \eqref{second representation of u} holds for any $(t,x) \in [0,T]\times \mathbb{R}_+$, by taking $x=b(t)$, $t \leq T$, on both sides of \eqref{second representation of u}, and by recalling that $u(t,b(t))=f(t)$, we find that $b$ solves
\begin{align}
\label{integralequation}
f(t) &=\mathbb{E}\bigg[ \int_0^{(T-t)\wedge S(b(t))} -f^{\prime}(t+\theta)\mathbbm{1}_{\{b(t) + \mu \theta + \sigma W_{\theta} \geq b(t+\theta)\}}~d\theta  \nonumber \\
&+ m(t+S(b(t)))\mathbbm{1}_{\{S(b(t)) \leq T-t\}}+g_x(T,A_{T-t}(b(t))) \bigg]. 
\end{align}
By adapting arguments as those in Section 25 of Peskir and Shiryaev \cite{Peskir2006}, based on the superharmonic characterization of $u$, one might then prove that $b$ is the unique solution to \eqref{integralequation} among a suitable class of continuous and positive functions. 
\end{remark}

The next result follows from \eqref{second representation of u} by expressing the expected value as an integral with respect to the probability densities of the involved processes and random variables. Its proof can be found in the Appendix for the sake of completeness.

\begin{corollary}\label{smooth fit for u}
The function $u(t,\cdot)$ is continuously differentiable on $(0,\infty)$ for all $t \in [0,T)$.
\end{corollary}

In the next section we will suitably integrate the two alternative representations of $u$ \eqref{OST with optimal stopping time} and \eqref{second representation of u} with respect to the space variable, and we will show that such integrations give the value function \eqref{value function of ocp} of the optimal dividend problem. As a byproduct, we will also obtain the optimal dividend strategy $D^{\star}$.

%%%%%%%%%%%%%%%%%%%%%%%%%%%%%%%%%%%%%%%%%%%%%%%%%%%%%%%%%%%%%%%%%%%%%%%%%%%%%%%%%%%%

\subsection{Integrating the Optimal Stopping Value Function}

In the next two propositions we integrate with respect to the space variable the two representations of $u$ given by \eqref{OST with optimal stopping time} and \eqref{second representation of u}. The proofs will employ pathwise arguments. However, in order to simplify exposition, we will not stress the $\omega$-dependence of the involved random variables and processes.

\begin{proposition}
Let $b$ the optimal stopping boundary of problem \eqref{Stopping problem}, recall 
\[I^0_s(x)=\max_{0 \leq \theta \leq s}\{ -x-\mu \theta - \sigma W_{\theta}\} \vee 0, \quad s \geq 0,\]
and define 
\begin{align}
R_s(x):= x + \mu s +\sigma W_s + I^0_s(x), \quad s \geq 0.  \label{RBM without boundary}
\end{align}

Then for any $(t,x) \in [0,T]\times \mathbb{R}_+$ one has
\begin{equation}\label{integration of u with N}
\int_x^{b(t)} u(t,y)~dy=N(t,b(t))-N(t,x),
\end{equation}
where 
\begin{align}
N(t,x)
:=\mathbb{E}\bigg[&-\int_0^{T-t} \big(R_s(x)-b(t+s)\big)^+ f^{\prime}(t+s)~ds - \int_0^{T-t}m(t+s) ~dI^0_s(x) \nonumber \\
&+  g(T,R_{T-t}(x))\bigg]. \label{definition of N}
\end{align}
\end{proposition}

\begin{proof}
To prove \eqref{integration of u with N} we use representation \eqref{second representation of u} of the value function of the optimal stopping problem \eqref{Stopping problem}. 
Using Fubini-Tonelli's Theorem we obtain
\begin{align}
\int_x^{b(t)} u(t,y)~dy
&=\int_x^{b(t)} \mathbb{E}\bigg[ \int_0^{(T-t)\wedge S(y)} -f^{\prime}(t+s)\mathbbm{1}_{\{y + \mu s + \sigma W_{s} \geq b(t+s)\}}~ds  \nonumber \\
&+ m(t+S(y))\mathbbm{1}_{\{S(y) \leq T-t\}}+g_x(T,A_{T-t}(y)) \bigg] ~dy \nonumber \\
&= \mathbb{E}\bigg[ -\int_0^{(T-t)} f^{\prime}(t+s) \bigg( \int_x^{b(t)}\mathbbm{1}_{\{y + \mu s + \sigma W_{s} \geq b(t+s)\}}\mathbbm{1}_{\{ s \leq S(y) \}}~dy \bigg)~ds  \label{integral of u with 3 summands} \\
&+ \int_x^{b(t)} m(t+S(y))\mathbbm{1}_{\{S(y) \leq T-t\}}~dy+\int_x^{b(t)}g_x(T,A_{T-t}(y))~dy \bigg]. \nonumber
\end{align}
In the following we investigate separately the three summands of the last term on the right-hand side of \eqref{integral of u with 3 summands}.

Recalling $S(x)=\inf \{ u \geq 0 : x +\mu u + \sigma W_u =0 \}$ it is clear that
\begin{equation} \label{S and M inverse 1}
S(y)\geq s \Leftrightarrow M_s \leq y
\end{equation}
for any $(s,y) \in \mathbb{R}_+ \times (0,\infty)$, where we have defined 
\begin{equation}
M_s:=\max_{0\leq \theta \leq s}(-\mu \theta - \sigma W_{\theta}), \quad s \geq 0. \label{def of maximum of RBM}
\end{equation}
We can then rewrite \eqref{RBM without boundary} in terms of \eqref{def of maximum of RBM} and obtain
\begin{equation}
R_s(x)=(x \vee M_s) + \mu s + \sigma W_s, \quad s \geq 0. \label{RBM 2}
\end{equation}
By using \eqref{S and M inverse 1} we find

\begin{align}
&\int_x^{b(t)} \mathbbm{1}_{\{y + \mu s + \sigma W_s \geq b(t+s)\}}\mathbbm{1}_{\{S(y)\geq s\}} ~dy =\int_{x \vee \big[b(t+s)-\mu s - \sigma W_s\big]}^{b(t) \vee \big[b(t+s)-\mu s - \sigma W_s\big]} \mathbbm{1}_{\{S(y)\geq s\}} ~dy \nonumber \\
&= \int_{x \vee \big[b(t+s)-\mu s - \sigma W_s\big]}^{b(t) \vee \big[b(t+s)-\mu s - \sigma W_s\big]} \mathbbm{1}_{\{M_s \leq y\}} ~dy \nonumber \\
&=\big[(b(t) \vee (b(t+s)-\mu s - \sigma W_s) \vee M_s)-(x\vee (b(t+s)-\mu s - \sigma W_s) \vee M_s)\big] \nonumber \\
&=\big[(b(t) \vee M_s) \vee (b(t+s)-\mu s - \sigma W_s) -(x \vee M_s) \vee (b(t+s)-\mu s - \sigma W_s)\big] \label{integration of f for N} \\
&=\big[\big([(b(t) \vee M_s)+\mu s + \sigma W_s ]\vee b(t+s)\big)-\big([(x\vee M_s)+\mu s + \sigma W_s ]\vee b(t+s)\big)\big] \nonumber \\
&=\big[\big(R_s(b(t)) \vee b(t+s)\big)-\big(R_s(x) \vee b(t+s)\big)\big] \nonumber \\
&=\big[\big(R_s(b(t)) - b(t+s)\big)^+ -\big(R_s(x) - b(t+s)\big)^+\big]. \nonumber
\end{align}

For the third summand of the last term of the right-hand side of \eqref{integral of u with 3 summands} we have, due to the fact that $g_x(T,\Delta)=0$,
\begin{align}
\int_x^{b(t)} g_x(T,A_{T-t}(y)) dy 
&=\int_x^{b(t)} g_x(T,y+ \mu (T-t) + \sigma W_{T-t}) \mathbbm{1}_{\{S(y)>T-t\}} dy \nonumber \\
&=\int_x^{b(t)} g_x(T,y+ \mu (T-t) + \sigma W_{T-t}) \mathbbm{1}_{\{M_{T-t} < y\}} dy \label{integration of g for N} \\
&=\int_{x\vee M_{T-t}}^{b(t) \vee M_{T-t}} g_x(T,y+ \mu (T-t) + \sigma W_{T-t})  dy \nonumber \\
&=g(T,R_{T-t}(b(t)))-g(T,R_{T-t}(x)), \nonumber
\end{align}
where in the last step we use \eqref{RBM 2}.
To prove that 
\begin{align}
\int_x^{b(t)} m(t+S(y))\mathbbm{1}_{\{S(y) \leq T-t\}}dy=  \int_0^{T-t} m(t+s) dI^0_s(x) - \int_0^{T-t} m(t+s) dI^0_s(b(t)) \label{integration of m without boundary}
\end{align}
we have to distinguish two cases. In the following we let $(t,x) \in [0,T] \times \mathbb{R}_+$ be given and fixed, and we prove \eqref{integration of m without boundary} by taking $x < b(t)$. The arguments are exactly the same if $b(t) < x$ by reversing the roles of $x$ and $b(t)$.
\vspace{0.25cm}

\textbf{Case 1.}\, Here we take $x \in \{ y \in \mathbb{R}_+: S(y) \geq T-t \}$; that is, the initial point $x>0$ is such that the drifted Brownian motion is not reaching $0$ before the time horizon. This implies that $R_s(x)$ in \eqref{RBM without boundary} equals $x + \mu s + \sigma W_s$ and so $I^0_s(x)=0$ for all $s \in [0,T-t]$. Hence, we can write
\begin{align}
\int_x^{b(t)} m(t+S(y))\mathbbm{1}_{\{S(y)\leq T-t\}} dy =0=  \int_0^{T-t} m(t+s) dI^0_s(x) - \int_0^{T-t} m(t+s) dI^0_s(b(t)),
\label{integration of m without boundary case 1}
\end{align}
where we have used that $S(y)> S(x) \geq T-t$ for any $y > x$ and $\{x\}$ has zero Lebesgue measure to obtain the first equality, and the fact that $0=I_s^0(x)\geq I_s^0(b(t))\geq 0$ since $x<b(t)$.
\vspace{0.25cm}

\textbf{Case 2.}\, Here we take $x \in \{ y \in \mathbb{R}_+: S(y) < T-t \}$; i.e., the drifted Brownian motion reaches $0$ before the time horizon.
Define 
\begin{equation}
\label{def:zCase2}
z:= \inf\{ y \in \mathbb{R}_+:S(y) \geq T-t \},
\end{equation}
with the usual convention $\inf \emptyset = + \infty$. In the sequel we assume that $z<+\infty$, since otherwise there is no need for the following analysis to be performed.
Note that, by continuity in time and in the initial datum of the paths of the drifted Brownian motion, we have $S(z)\leq T-t$. Furthermore, it holds for all $y \in [x,z]$ that (cf.\ \eqref{def of maximum of RBM})
\begin{equation}
y+I_s^0(y)=M_s, \quad \forall s \geq S(y), \label{property 1 in case 1}
\end{equation}
\begin{equation}
I^0_s(y)=0, \quad \forall s < S(y). \label{property 2 in case 1}
\end{equation}

Using \eqref{property 1 in case 1}, \eqref{property 2 in case 1}, \eqref{S and M inverse 1}, and the change of variable formula in Section 4 of Chapter 0 of the book by Revuz and Yor \cite{revuz2013} (see also equation (4.7) in Baldursson and Karatzas \cite{Baldursson1996}) we obtain

\begin{align}
&\int_x^{z\wedge b(t)} m(t+S(y))\mathbbm{1}_{\{ S(y) \leq T-t \}} dy 
=\int_x^{z\wedge b(t)} m(t+S(y)) dy \nonumber \\
&=\int_{S(x)}^{S({z\wedge b(t)})} m(t+s) dM_s =\int_{S(x)}^{S({z\wedge b(t)})} m(t+s) \left(dI_s^0(x)-dI_s^0({z\wedge b(t)})\right)) \label{integration of m without boundary case 2} \\
&=\int_{0}^{T-t} m(t+s) \left(dI_s^0(x)-dI_s^0({z\wedge b(t)})\right) \nonumber \\
&=\int_{0}^{T-t} m(t+s) dI_s^0(x) - \int_{0}^{T-t} m(t+s)dI_s^0({z\wedge b(t)}).\nonumber
\end{align}
For the integral $\int_{z \wedge b(t)}^{b(t)} m(t+S(y))\mathbbm{1}_{\{S(y) \leq T-t\}}~dy$ we can use the result of Case 1 due to the definition of $z$ \eqref{def:zCase2}. Then, combining \eqref{integration of m without boundary case 1} and \eqref{integration of m without boundary case 2} leads to \eqref{integration of m without boundary}. \\Ù
\newline
By \eqref{integration of f for N}, \eqref{integration of g for N} and \eqref{integration of m without boundary}, and recalling \eqref{definition of N} and \eqref{integral of u with 3 summands} we obtain \eqref{integration of u with N}.
\end{proof}

\begin{proposition}
Let $(D^{\star},I^{\star})$ be the solution to system \eqref{system for I and D}. Then, for any $(t,x) \in [0,T]\times \mathbb{R}_+$ one has

\begin{equation}\label{integration of u with M}
\int_x^{b(t)} u(t,y)~dy=M(t,b(t))-M(t,x),
\end{equation}
where $b$ is the optimal stopping boundary of problem \eqref{Stopping problem} and
\begin{equation}\label{def of M}
M(t,x):=\mathbb{E}\left[\int_0^{T-t} f(t+s) ~dD^{\star}_s(t,x)-\int_0^{T-t}m(t+s)~dI^{\star}_s(t,x)+ g(T,X^{D^{\star}}_{T-t}(x))\right].
\end{equation}
\end{proposition}

\begin{proof}
For this proof we use instead the representation of $u$ (cf.\ \eqref{OST with optimal stopping time})
%\begin{align}
%u(t,x)
%&=  \mathbb{E}\Big[f(t+\tau^{\star}(t,x)) \mathbbm{1}_{\left\{A_{\tau^{\star}(t,x)}(x)>0  \right\}} \mathbbm{1}_{\left\{\tau^{\star}(t,x) < T-t \right\}} +m(t+S(x))  \mathbbm{1}_{\left\{A_{\tau^{\star}(t,x)}(x) =\Delta \right\}} \nonumber \\
%&+  g_x(T,A_{T-t}(x)) \mathbbm{1}_{\left\{\tau^{\star}(t,x)=T-t\right\}}  \Big]
%\end{align}
%with $\tau^{\star}(t,x)$ as in \eqref{optimal stopping time}. To simplify the proof notice that, due to \eqref{absorbed process}, $\mathbbm{1}_{\left\{A_{u}(x)>0  \right\}}=\mathbbm{1}_{\left\{u<S(x)  \right\}}$ and $\mathbbm{1}_{\left\{A_{u}(x) =\Delta \right\}}=\mathbbm{1}_{\left\{u\geq S(x) \right\}}$ a.s. Using this, equation \eqref{OST with optimal stopping time} can be read as
\begin{align}
u(t,x)
&=\mathbb{E}\Big[f(t+\tau^{\star}(t,x)) \mathbbm{1}_{\left\{\tau^{\star}(t,x) < T-t \wedge S(x) \right\}} +m(t+S(x))  \mathbbm{1}_{\left\{\tau^{\star}(t,x)\geq S(x) \right\}} \nonumber \\
&+  g_x(T,A_{T-t}(x)) \mathbbm{1}_{\left\{\tau^{\star}(t,x)=T-t < S(x)\right\}} \Big].
\end{align}

The proof is quite long and technical and it is organized in four steps. Moreover, in order to simplify exposition from now we set $t=0$. Indeed, all the following arguments remain valid if $t \in (0,T]$ by obvious modifications. \vspace{0.25cm} \\ 
If $x \geq b(0)$, then \eqref{integration of u with M} clearly holds. Indeed, $\int_x^{b(0)} u(0,y) ~dy=-(x-b(0))f(0)$ since $\tau^{\star}(0,y)=0$ for any $y \geq b(0)$. Also, from \eqref{def of M} $M(0,b(0))-M(0,x)=M(0,b(0))-\big[(x-b(0))f(0)+M(0,b(0))\big]$, since $D^{\star}(0,x)$ has an initial jump of size $(x-b(0))$ which is such that $X_{0+}^{D^{\star}}(x)=b(0)$. Hence, in the following we prove \eqref{integration of u with M} assuming that $x < b(0)$.

\vspace{0.25cm}

\textbf{Step 1.}\,
Here we take $x \in \{y\in \mathbb{R}_+: \tau^{\star}(0,y) < S(y) \}$; that is, the initial point $x>0$ is such that either the drifted Brownian motion reaches the boundary before hitting the origin, or the time horizon arises before hitting the origin.
Define the process $(L_s)_{s\geq 0}$ such that
\begin{equation}
L_s:=\max_{0\leq \theta \leq s}\{\mu \theta + \sigma W_{\theta} - b(\theta)\},\qquad 0 \leq s \leq T. \label{pocess L}
\end{equation}
Then we have that for all $y \in [x,b(0)]$
\begin{equation}
\{ \tau^{\star}(0,y) \leq s \} = \{L_s\geq -y\}, \label{L inverse of optimal stopping time}
\end{equation}
\begin{equation}
\{ \tau^{\star}(0,y) =T \} = \{L_T \leq -y\}, \label{L inverse of optimal stopping time 2}
\end{equation}

\begin{equation}\label{Dstar in case 1}
D^{\star}_s(0,y)=\begin{cases}
0, & 0\leq s \leq \tau^{\star}(0,y), \\
y+L_s, & \tau^{\star}(0,y)\leq s \leq S(y),
\end{cases}
\end{equation}
and 
\begin{equation}
X^{D^{\star}}_s(y)=\begin{cases}
y+\mu s+ \sigma W_s, & 0\leq s \leq \tau^{\star}(0,y), \\
\mu s+ \sigma W_s-L_s, & \tau^{\star}(0,y)\leq s \leq S(y),
\end{cases}
\end{equation}
and in particular (cf.\ \eqref{system for I and D}) $I^{\star}_s(0,y)=I^{\star}_s(0,b(0))=0$ for any $s \in [0,\tau^{\star}(0,y)]$.

Moreover, it follows by definition of $\tau^{\star}(0,x)$, $S(x)$ and $X^{D^{\star}}(x)$ that for all $y \in [x,b(0)]$ we have 
\begin{equation}\label{1a}
0=\tau^{\star}(0,b(0))\leq \tau^{\star}(0,y)\leq \tau^{\star}(0,x),
\end{equation}

%\begin{equation}\label{1b}
%S(x)=S(y)=S(b(0)),
%\end{equation}

\begin{equation}\label{1c}
\tau^{\star}(0,y)<\tau^{\star}(0,x)<S(x) \leq S(y),
\end{equation}
and
\begin{equation}\label{1d}
\text{ on $\left\{\tau^{\star}(0,x) <T\right\}$:} \quad X^{D^{\star}}_s(y)=X^{D^{\star}}_s(x),  \quad \forall s > \tau^{\star}(0,x).
\end{equation}
With these results at hand, we now show that for all $x \in [0,b(0)]$ such that $\tau^{\star}(0,x)<S(x)$ it holds that
\begin{equation}
\int_x^{b(0)}f(\tau^{\star}(0,y))1_{\{\tau^{\star}(0,y)< S(y)\}}dy=\int_0^{T}f(s)~d D^{\star}_s(0,b(0)) -\int_0^{T}f(s)~d D^{\star}_s(0,x),\label{integration optimal risk case 1 for f}
\end{equation}
\begin{equation}
\int_x^{b(0)} g_x (T,y+\mu T +\sigma W_T)\mathbbm{1}_{\left\{\tau^{\star}(0,y)=T< S(y)\right\}}~dy=g(T,X^{D^{\star}}_T(b(0)))-g(T,X^{D^{\star}}_T(x))
\label{integration optimal risk case 1 for g}
\end{equation}
and
\begin{equation}
\int_x^{b(0)}m(S(y))1_{\{\tau^{\star}(0,y)\geq S(y)\}}dy=\int_0^{T}m(s)~d I^{\star}_s(0,x) -\int_0^{T}m(s)~d I^{\star}_s(0,b(0)).\label{integration optimal risk case 1 for m}
\end{equation}
We start with \eqref{integration optimal risk case 1 for f}. By \eqref{1d} we have that $dD^{\star}_s(0,x)=dD^{\star}_s(0,b(0))$ for all $\tau^{\star}(0,x) < s \leq T$.
By \eqref{Dstar in case 1},
and since $\tau^{\star}(0,b(0))=0$ one also has
\begin{equation}
D^{\star}_s(0,b(0))=b(0)+L_s, ~\forall s \in [0,S(b(0))]. \label{Dstar if starting at boundary}
\end{equation}
Hence the right-hand side of \eqref{integration optimal risk case 1 for f} rewrites as
\begin{equation}
\begin{aligned}
&\int_0^{T}f(s)~d D^{\star}_s(0,b(0)) -\int_0^{T}f(s)~d D^{\star}_s(0,x)
=\int_0^{\tau^{\star}(0,x)}f(s)~dD^{\star}_s(0,b(0)) \\
&- \int_0^{\tau^{\star}(0,x)}f(s)~dD^{\star}_s(0,x)
=\int_0^{\tau^{\star}(0,x)}f(s)~dD^{\star}_s(0,b(0))=\int_0^{\tau^{\star}(0,x)}f(s)~dL_s, \label{Step 1 Dividenden}
\end{aligned}
\end{equation}
where we have used that $dD^{\star}_s(0,x)=0$ for all $s \in [0,\tau^{\star}(0,x)]$ by \eqref{Dstar in case 1}.
However, by using a change of variable formula as in Baldursson and Karatzas \cite{Baldursson1996}, equation (4.7), we obtain
\begin{equation}
\begin{aligned}
\int_x^{b(0)}f(\tau^{\star}(0,y))1_{\{\tau^{\star}(0,y)< S(y)\}}dy
=\int_x^{b(0)}f(\tau^{\star}(0,y))dy=\int_0^{\tau^{\star}(0,x)}f(s)~dL_s, 
\end{aligned} \label{Step 1 Dividenden 2}
\end{equation}
where we have used \eqref{1c} in the first step, and the fact that $L_{ \cdot }$ is the left-continuous inverse of $\tau^{\star}(0,y)$ (cf.\ \eqref{L inverse of optimal stopping time}) in the last equality. Combining \eqref{Step 1 Dividenden} and \eqref{Step 1 Dividenden 2} equation \eqref{integration optimal risk case 1 for f} holds.

Next we show \eqref{integration optimal risk case 1 for g}. Using \eqref{Dstar if starting at boundary} and again \eqref{1d} we obtain for the right-hand side of \eqref{integration optimal risk case 1 for g} that
\begin{equation*}
g(T,X^{D^{\star}}_T(b(0)))-g(T,X^{D^{\star}}_T(x))=\left[g(T,\mu T +\sigma W_T-L_T)-g(T,x+\mu T+\sigma W_T)\right]\mathbbm{1}_{\{\tau^{\star}(0,x)=T\}}.
\end{equation*}
Also, \eqref{L inverse of optimal stopping time 2} and \eqref{1c} yields 
\begin{align}
&\int_x^{b(0)} g_x (T,y+\mu T +\sigma W_T)\mathbbm{1}_{\left\{\tau^{\star}(0,y)=T\right\}}~dy 
=\int_x^{b(0)} g_x (T,y+\mu T +\sigma W_T)\mathbbm{1}_{\left\{y\leq -L_T\right\}}~dy \nonumber \\
&=\left[g(T,\mu T +\sigma W_T-L_T)-g(T,x+\mu T+\sigma W_T)\right]\mathbbm{1}_{\{\tau^{\star}(0,x)=T\}}. \nonumber
\end{align}
Hence, we obtain \eqref{integration optimal risk case 1 for g}.

Finally, for \eqref{integration optimal risk case 1 for m}
there is nothing to show. In fact, the left-hand side is equal $0$ by \eqref{1c}, while the right-hand side is zero since the processes $I^{\star}(0,x)=I^{\star}(0,b(0))$ coincide (cf.\ \eqref{1d}).

\vspace{0.25cm}

\textbf{Step 2.}\,
Here we take $x \in \left\{y \in \mathbb{R}_+:\tau^{\star}(0,y)>S(y),\ \tau^{\star}(0,q)<S(q) ~\forall q \in (y,b(0))\right\}$. For a realization like that, such an $x$ is such that the drifted Brownian motion touches the origin before hitting the boundary, but it does not cross the origin. This in particular implies that $I_s^{\star}(0,x)=0$ for all $s \leq \tau^{\star}(0,x)$. Hence the same arguments employed in Step 1 hold true, and \eqref{integration optimal risk case 1 for f} -- \eqref{integration optimal risk case 1 for m} follow. 

\vspace{0.25cm}

\textbf{Step 3.}\,
Here we take $x \in \{ y \in \mathbb{R}_+: \tau^{\star}(0,y) > S(y) \}$; that is, the drifted Brownian motion hits the origin before reaching the boundary.

Define
\begin{equation}
z:=\inf\left\{y \in [0,b(0)]: \tau^{\star}(0,y)<S(y)\right\}\label{z in Case 3}
\end{equation}
which exists finite since $y \mapsto \tau^{\star}(0,y)-S(y)$ is decreasing and $\tau^{\star}(0,b(0))=0$ and $S(0)=0$ a.s.  We want to prove that

\begin{equation}
\int_x^z m(S(y)) \mathbbm{1}_{\{\tau^{\star}(0,y)\geq S(y)\}}~dy=\int_0^{T}m(s)~dI^{\star}_s(0,x)-\int_0^{T}m(s)~dI^{\star}_s(0,z),\label{integration optimal risk case 3 for m}
\end{equation}
\begin{equation}
\int_x^{z}f(\tau^{\star}(0,y))1_{\{\tau^{\star}(0,y)< S(y)\}}dy=\int_0^{T}f(s)~d D^{\star}_s(0,z) -\int_0^{T}f(s)~d D^{\star}_s(0,x),\label{integration optimal risk case 3 for f}
\end{equation}
and
\begin{align}
&\int_x^{z} g_x(T,y+\mu T +\sigma W_T)\mathbbm{1}_{\left\{\tau^{\star}(0,y)=T< S(y)\right\}}~dy \nonumber \\
&=\left[g(T,X^{D^{\star}}_T(z))-g(T,X^{D^{\star}}_T(x))\right].
\label{integration optimal risk case 3 for g}
\end{align}
Recall the process $(M_s)_{s \geq 0}$ of \eqref{def of maximum of RBM} such that
\[
M_s=\max_{0\leq \theta \leq s}(-\mu \theta - \sigma W_{\theta}), \quad s \geq 0,
\]
and (cf.\ \eqref{S and M inverse 1})
\[
\{ M_s \geq x \} = \{ S(x) \leq s \}  \quad \forall s \geq 0.
\]
For all $y \in [x,z)$ and $s \in [0,\tau^{\star}(0,y)]$ we have
\begin{equation}\label{L in case 3}
I^{\star}_s(0,y)= (M_s-y)^+ = \begin{cases}
0, & 0\leq t \leq S(y) \\
M_s-y, & S(y)\leq s \leq \tau^{\star}(0,y),
\end{cases}
\end{equation}
and
\begin{equation}\label{X in case 3}
X^{D^{\star}}_s(y)=\begin{cases}
y+\mu s+ \sigma W_s, & 0\leq s \leq S(y) \\
\mu s+ \sigma W_s+M_s, & S(y)\leq s \leq \tau^{\star}(0,y),
\end{cases}=(y \vee M_s)+\mu s + \sigma W_s.
\end{equation}
Also, it follows by \eqref{X in case 3} and \eqref{L in case 3} that for all $y \in [x,z)$
\begin{equation}\label{2d}
X^{D^{\star}}_s(y)=X^{D^{\star}}_s(z) ~\forall s \geq S(z).
\end{equation}
Moreover, recall that

\begin{equation}\label{2a}
S(x)\leq S(y) \leq S(z),
\end{equation}

\begin{equation}\label{2b}
\tau^{\star}(0,y)>S(y),
\end{equation}

%\begin{equation}\label{2c}
%T \wedge S(y) \leq \tau^{\star}(0,y),
%\end{equation}

With these observations at hand we can now show \eqref{integration optimal risk case 3 for m}-\eqref{integration optimal risk case 3 for g}.

By \eqref{2d} we have that $dI^{\star}_s(0,x)=dI^{\star}_s(0,z)$ for all $s \geq S(z)$. Further, we have that $I^{\star}_s(0,z)=0 $ for all $s \leq S(z)$. Therefore, by \eqref{2a} $I^{\star}_s(0,z)=I^{\star}_s(0,x)=0$ for $s\leq S(x)$, and the right-hand side of \eqref{integration optimal risk case 3 for m} rewrites as

\begin{align}
&\int_0^{T}m(s)~dI^{\star}_s(0,x)-\int_0^{T}m(s)~dI^{\star}_s(0,z)
=\int_{S(x)}^{S(z)}m(s)~[dI^{\star}_s(0,x)-dI^{\star}_s(0,z)] \nonumber \\
&=\int_{S(x)}^{S(z)}m(s)~dI^{\star}_s(0,x) =\int_{S(x)}^{S(z)}m(s)~dM_s. \label{Step 3 integration m 1}
\end{align}
Here we have used \eqref{L in case 3} with $y=x$.

On the other hand, for the left-hand side of \eqref{integration optimal risk case 3 for m}, we use the change of variable formula of Section 4 in Chapter 0 of Revuz and Yor \cite{revuz2013}. This leads to 
\begin{equation}
\begin{aligned}
\int_x^z m(S(y)) \mathbbm{1}_{\{\tau^{\star}(0,y)\geq S(y)\}}~dy=\int_x^z m(S(y)) ~dy =\int_{S(x)}^{S(z)} m(s)~dM_s,
\end{aligned} \label{Step 3 integration m 2}
\end{equation}
where we use \eqref{2b}, the fact that $\{z\}$ is a Lebesgue zero set, and that $M$ is the right-continuous inverse of $S$ (see \eqref{S and M inverse 1}).
Combining \eqref{Step 3 integration m 1} and \eqref{Step 3 integration m 2} proves \eqref{integration optimal risk case 3 for m}.

Equation \eqref{integration optimal risk case 3 for f} follows by observing that \eqref{2d}--\eqref{2a} imply that the processes $D^{\star}(0,z)$ and $D^{\star}(0,x)$ coincide, and the left-hand side equals $0$ by definition. Notice that for such an argument particular care has to be put when considering $z$ of \eqref{z in Case 3} as a starting point for the drifted Brownian motion. In particular, if the realization of the Brownian motion is such that $\tau^{\star}(0,z) < S(z)$, then by definition of $z$, the drifted Brownian motion only touches the boundary at time $\tau^{\star}(0,z)$, but does not cross it. Hence, we still have $D_s^{\star}(0,z)=0$ for all $s \leq S(z)$, which implies \eqref{2d} and therefore still $D_s^{\star}(0,z)=D_s^{\star}(0,x)$. In turn, this gives again that \eqref{integration optimal risk case 3 for f} holds also for such a particular realization of the Brownian motion.

Finally, to prove equation \eqref{integration optimal risk case 3 for g} remember that $x \in \left\{y \in \mathbb{R}_+:\tau^{\star}(0,y)>S(y)\right\}$. By definition of $z$ we obtain $\tau^{\star}(0,y)\geq S(y)$ for all $y \in [x,z)$ and the left-hand side of \eqref{integration optimal risk case 3 for g} equals zero. By \eqref{2d} the processes $X^{D^{\star}}_s(z)=X^{D^{\star}}_s(x)$ coincides for all $s \geq S(z)$, and $S(z)\leq T$ a.s.\ by Lemma \ref{Lemma S(z) smaller T} in the Appendix. Therefore, the right-hand side of \eqref{integration optimal risk case 3 for g} equals zero as well.

\vspace{0.25cm}

\textbf{Step 4.}\,
For $x \in \{y\in \mathbb{R}_+: \tau^{\star}(0,y)<S(y)\}$, \eqref{integration of u with M} follows by the results of Step 1. If, instead, $x \in \{y\in \mathbb{R}_+: \tau^{\star}(0,y)>S(y)\}$, then we can integrate $u$ separately in the intervals $[x,z]$ and $[z,b(0)]$. When integrating $u$ in the interval $[x,z]$ we use the results of Step 3. On the other hand, integrating $u$ over $[z,b(0)]$ we have to distinguish two cases. Now, if $z$ belongs to $\{y\in \mathbb{R}_+: \tau^{\star}(0,y)<S(y)\}$, then we can still apply the results of Step 1 to conclude. If $z$ belongs to $\left\{y \in \mathbb{R}_+:\tau^{\star}(0,y)>S(y),\ \tau^{\star}(0,q)<S(q) ~\forall q \in (y,b(0))\right\}$, we can employ the results of Step 2 to obtain the claim. Thus, in any case, \eqref{integration of u with M} holds. 
\end{proof}

We now prove that the two functions $N$ and $M$ of \eqref{definition of N} and \eqref{def of M}, respectively, are such that $N=M$. To accomplish that we preliminary notice that by their definitions and strong Markov property, one has that the processes

\begin{equation}\label{N - int dl is martingale}
N(t+s \wedge \tau^{\star}(t,x), R_{s \wedge \tau^{\star}(t,x)}(x)) -
\int_0^{s \wedge \tau^{\star}(t,x)} m(t+\theta)~dI^0_{\theta}(x), ~~0\leq s \leq T-t,
\end{equation}
and
\begin{equation}\label{M - int dl is martingale}
M(t+s \wedge \tau^{\star}(t,x), R_{s \wedge \tau^{\star}(t,x)}(x)) -
\int_0^{s \wedge \tau^{\star}(t,x)} m(t+\theta)~dI^{\star}_{\theta}(t,x), ~~0\leq s \leq T-t,
\end{equation}
are $\mathbb{F}$-martingales for any $(t,x) \in [0,T] \times \mathbb{R}_+$. Moreover, by \eqref{integration of u with N} one has $N(t,x)=N(t,b(t))-\int_x^{b(t)}u(t,y)~dy$ and, due to \eqref{integration of u with M}, $M(t,x)=M(t,b(t))-\int_x^{b(t)}u(t,y)~dy$.
Hence,
\begin{equation}
\Psi(t):= M(t,x)-N(t,x),\quad t \in [0,T],
\end{equation}
is independent of the $x$ variable. We now prove that one actually has $\Psi=0$ and therefore $N=M$.

\begin{theorem}\label{N=M}
It holds $\Psi(t)=0$ for all $t \in [0,T]$. Therefore, $N=M$ on $[0,T]\times \mathbb{R}_+$. 
\end{theorem}

\begin{proof} 
Since $(N-M)$ is independent of $x$, it suffices to show that $(N-M)(t,x)=0$ at some $x$ for any $t\leq T$. To accomplish that we show $\Psi^{\prime}(t)=0$ for any $t<T$, since by \eqref{integration of u with N} and \eqref{integration of u with M} we already know that
\[\Psi(T)=N(T,x)-M(T,x)=g(T,x)-g(T,x)=0.\]

Then take $ 0 < x_1 < x_2, t_0 \in [0,T)$ and $\varepsilon>0$ such that $t_0+\varepsilon <T$ given and fixed, consider the rectangular domain $\mathcal{R}:=(t_0-\varepsilon,t_0+\varepsilon)\times (x_1,x_2)$ such that $cl(\mathcal{R}) \subset \mathcal{C}$ (where $\mathcal{C}$ has been defined in \eqref{continuation region}). Also, denote by $\partial_0 \mathcal{R}:= \partial \mathcal{R} \backslash \left( \{t_0-\varepsilon\}\times (x_1,x_2)\right)$. Then consider the problem 
\begin{equation*}
(P)\, 
\begin{cases}
h_t(t,x) = \mathcal{L}h(t,x), \quad (t,x) \in \mathcal{R},\\
h(t,x)=(N-M) (t,x), \quad (t,x) \in \partial_0 \mathcal{R},
\end{cases}
\end{equation*}
where $\mathcal{L}$ is the second-order differential operator that acting on $\varphi \in C^{1,2}([0,T]\times \mathbb{R})$ gives
\[ (\mathcal{L}\varphi) (t,x)= \mu \frac{\partial \varphi}{\partial x}(t,x) +\frac{1}{2}\sigma^2 \frac{\partial^2 \varphi}{\partial x^2}(t,x), ~(t,x) \in [0,T] \times \mathbb{R}. \]

By reversing time, $t \mapsto T-t$, Problem (P) corresponds to a classical initial value problem with uniformly elliptic operator (notice that $\sigma^2>0$) and parabolic boundary $\partial_0 \mathcal{R}$. Since $N-M$ is continuous, and all the coefficients in the first equation of $(P)$ are smooth (actually constant), by classical theory of partial differential equations of parabolic type (see, e.g., Chapter V in the book by Lieberman \cite{lieberman1996second}) problem $(P)$ admits a unique solution $h$ that is continuous, with continuous derivatives $h_t,h_x,h_{xx}$. Moreover, by the Feynman-Kac's formula, such a solution admits the representation

\[h(t,x)=\mathbb{E}[(N-M)(t+\widehat{\tau}(t,x),Z_{\widehat{\tau}(t,x)}(x))], \]
where 
\[\widehat{\tau}(t,x) := \inf \{ s \in [0,T-t):(t+s,Z_s(x)) \in \partial_0 \mathcal{R} \} \wedge (T-t), \]
and $Z_s(x)=x+\mu s + \sigma W_s, s \geq 0$.
Notice that we have $\widehat{\tau}(t,x) \leq \tau^{\star}(t,x)$ a.s., since $cl(\mathcal{R}) \subset \mathcal{C}$. Also, the integral terms in \eqref{N - int dl is martingale} and \eqref{M - int dl is martingale} are equal since $dI^0_{\theta}(x)=dI_{\theta}^{\star}(t,x)=0$ for any $\theta \leq \widehat{\tau}(t,x)\leq \tau^{\star}(t,x)$. Hence by the martingale property of \eqref{N - int dl is martingale} and \eqref{M - int dl is martingale} we have 
\begin{equation}
h(t,x)=(N-M)(t,x)\, \text{ in } \mathcal{R},
\end{equation}
and, by arbitrariness of $\mathcal{R}$, 
\[\Psi(t)=(N-M)(t,x)=h(t,x)\, \text{ in } \mathcal{C}.\]
Therefore, since $\Psi=N-M$ is independent of $x$, continuous in $t$ and solves the first equation of $(P)$ in $\mathcal{C}$, we obtain
$\Psi^{\prime}(t)=0$ for any $t < T$. Hence $\Psi(t)=0$ for any $t \leq T$ since $\Psi(T)=0$, and thus $N(t,x)-M(t,x)=0$ for any $t \leq T$ and for any $x \in (0,\infty)$.

\end{proof}

In the following we show that the function $N$ is an upper bound for the value function $V$ of \eqref{value function of ocp}. We first prove the following result.

\begin{theorem}\label{Theorem for N}
For any $(t,x) \in \mathbb{R}_+ \times [0,T]$ the process 
\begin{equation}\label{N is supermartingale}
\widetilde{N}_s:=N(t+s, R_s(x)) -
\int_0^{s} m(t+u)~dI^0_{u}(x), \quad 0\leq s \leq T-t,
\end{equation} 
is an $\mathbb{F}$-supermartingale.
\end{theorem}

\begin{proof}
It is enough to show that $\mathbb{E}\big[\widetilde{N}_\theta\big] \leq \mathbb{E}[\widetilde{N}_\tau]$ for all bounded $\mathbb{F}$-stopping times $\theta,\tau$ such that $\theta \geq \tau$ (see Karatzas and Shreve \cite{karatzas1991brownian}, Chapter 1, Problem 3.26).

By the strong Markov property and the definition of $N$ \eqref{definition of N}, we get that for any bounded $\mathbb{F}$-stopping time $\rho$ one has
\begin{align*}
\mathbb{E}[\widetilde{N}_\rho ]
&=  \mathbb{E}\left[N(t+\rho,R_{\rho}(x)) -
\int_0^{\rho} m(t+s)~dI^0_{s}(x)\right] \\
&= \mathbb{E} \left[ -\int_{\rho}^{T-t} f^{\prime} (t+s)[R_s(x)-b(t+s)]^+ ds \right.- \left.\int_{0}^{T-t} m(t+s)~dI^0_{s}(x) 
+  g(R_{T-t}(x))\right] \\
&=N(t,x)+\mathbb{E}\left[\int_0^{\rho} f^{\prime} (t+s)\big(R_s(x)-b(t+s)\big)^+ ds \right] =: N(t,x)+\Delta_{\rho},
\end{align*}
for any $(t,x)\in [0,T] \times \mathbb{R}_+$. Hence, taking $\theta, \tau$ such that $T-t\geq \theta \geq \tau$ we get from the latter that $\mathbb{E}[\widetilde{N}_\theta ]=N(t,x)+\Delta_{\theta} \leq N(t,x)+\Delta_{\tau}=\mathbb{E}[\widetilde{N}_\tau ]$, where the inequality is due to the fact that $f^{\prime} \leq 0$ on $\mathcal{S}$ (cf.\ Corollary \ref{corollary for D}-(ii)). This proves the claimed supermartingale property.
\end{proof}

To proceed further, we need the following properties of the function $N$ of \eqref{definition of N}. Its proof is relegated to the Appendix.

\begin{lemma}\label{Lemma for Ito for N}
The function $N \in C^{1,2}([0,T)\times (0,\infty)) \cap C^0([0,T]\times \mathbb{R}_+)$.
\end{lemma}

Thanks to Lemma \ref{Lemma for Ito for N}, an application of It\^{o}'s formula allows us to obtain the following (unique) Doob-Meyer decomposition of the $\mathbb{F}$-supermartingale $\widetilde{N}$ (cf.\ \eqref{N is supermartingale}).

\begin{corollary}
The $\mathbb{F}$-supermartingale $\widetilde{N}$ of \eqref{N is supermartingale} is such that for all $(t,x) \in [0,T] \times \mathbb{R}_+$ and $s \in [0,T-t]$
\begin{equation}\label{Doob Meyer for N}
N(t+s, R_s(x)) -
\int_0^{s} m(t+\theta)~dI^0_{\theta}(x)=N(t,x)+ \sigma \int_0^s u(t+\theta, R_{\theta}(x))~dW_\theta + \Pi_s(t,x),
\end{equation} 
where $\Pi_{\cdot}(t,x)$ is a continuous, nonincreasing and $\mathbb{F}$-adapted process.
\end{corollary}

\begin{proof}
By the Doob-Meyer decomposition, the $\mathbb{F}$-supermartingale in \eqref{N is supermartingale} can be (uniquely) written as the sum of an $\mathbb{F}$-martingale and a continuous, $\mathbb{F}$-adapted nonincreasing process $(\Pi_s)_{s\geq0}$. Applying the martingale representation theorem to the martingale part of $\widetilde{N}$, yields the decomposition
\begin{equation}
\widetilde{N}_s= N(t,x) + \int_0^s \phi_{\theta} ~dW_{\theta} + \Pi_s(t,x),
\end{equation}
for some $\phi \in L^2(\Omega \times [0,T], \mathbb{P}\otimes dt)$. Finally, an application of It\^{o}'s lemma shows that $\phi_{\theta}=\sigma u(t+\theta,R_{\theta}(x))$ a.s. 
\end{proof}

\begin{theorem}\label{theorem for Q}
For any process $D\in \mathcal{D}(t,x)$ and any $(t,x) \in [0,T] \times \mathbb{R}_+$, the process
\begin{equation}\label{def of Q}
Q_s(D;t,x) := \int_{[0,s]}f(t+\theta)~dD_{\theta} -
\int_0^{s} m(t+\theta)~dI^{D}_{\theta} + N(t+s,X^D_{s}(x)),
\end{equation}  
$s\in [0,T-t]$, is such that
\begin{equation}\label{Q smaller N}
\mathbb{E}\left[Q_s(D;t,x)\right] \leq N(t,x), \quad \text{ for any } s \in [0,T-t].
\end{equation}
\end{theorem}

\begin{proof}
The proof is organized in 3 steps. 
\vspace{0.25cm}

\textbf{Step 1.}\, For $D\equiv0$, the proof is given by Theorem \ref{Theorem for N}.
\vspace{0,25cm}

\textbf{Step 2.}\, Let $D_s:=\int_0^s z_u ~du$, $s \geq 0$, where $z$ is a bounded, nonnegative, $\mathbb{F}$-progressively measurable process. To show \eqref{Q smaller N} we use Girsanov's Theorem and we rewrite the state process $X^{D}_s(x)=x +\mu s +\sigma W_s +D_s-I^{D}_s$ as a new drifted Brownian motion reflected at the origin.
We therefore introduce the exponential martingale
\[Z_s=\exp \left( \int_0^s \frac{z_u}{\sigma} ~dW_u - \frac{1}{2 \sigma^2} \int_0^s z_u^2~du \right), \quad s \geq 0, \]
and we obtain that under the measure $\widehat{\mathbb{P}}=Z_{T}\mathbb{P}$, the process
\[\widehat{W}_s:=W_s -\frac{1}{\sigma} \int_0^s z_u du, \quad s \geq 0,\]
is an $\mathbb{F}$- Brownian motion.

We can now rewrite the process $Q$ of \eqref{def of Q} under $\widehat{\mathbb{P}}$ as
\begin{equation}\label{Q under P hat}
Q_s(D;t,x)= \int_{[0,s]}f(t+\theta)~dD_{\theta} -
\int_0^{s} m(t+\theta)~d\widehat{I}^D_{\theta} + N(t+s,\widehat{R}_{s}(x)),
\end{equation}
for any $ s \in [0,T-t]$, where under $\widehat{\mathbb{P}}$
\[\widehat{X}^{D}_s(x)=x+\mu s + \sigma \widehat{W}_s + \widehat{I}^{D}_s=:\widehat{R}_s(x). \] 
Here $\widehat{I}_{\cdot}^{D}$ is flat off $\{ s\geq 0: \widehat{R}_s(x)=0 \}$ and reflects the drifted Brownian motion at the origin. By employing \eqref{Doob Meyer for N}, equation \eqref{Q under P hat} reads as 
\begin{align}
Q_s(D;t,x)
&=N(t,x)+\sigma \int_0^s u(t+u,\widehat{R}_u(x)) d\widehat{W}_u +\widehat{\Pi}_s(t,x), \quad s \in [0,T-t], \label{Q with martingale representation}
\end{align}
where we have set
\begin{equation}
\widehat{\Pi}_s(t,x):=\Pi_s(t,x)+\int_0^s \bigg(f(t+\theta)-u(t+\theta,R_{\theta}(x))\bigg)z_{\theta} d\theta, \quad s \in [0,T-t].
\end{equation}
Since $\widehat{\Pi}$ is nonincreasing due to the fact that $u \geq f$ and $\Pi_{\cdot}(t,x)$ is nonincreasing, we can take expectations in \eqref{Q with martingale representation} so to obtain
\[\mathbb{E}\left[Q_s(D;t,x) \right] \leq N(t,x), \quad \forall s \in [0,T-t].\]
\vspace{0,25cm}

\textbf{Step 3.}\, Since any arbitrary $D\in \mathcal{D}(t,x)$ can be approximated by an increasing sequence $(D^n)_{n \in \mathbb{N}}$ of absolutely continuous processes as the ones considered in Step 2 (see El Karoui and Karatzas \cite{Karoui1988}, Lemmata 5.4, 5.5 and Proposition 5.6), we have for all $n  \in \mathbb{N}$ 
\[ \mathbb{E}\left[Q_s(D^n;t,x)\right] \leq N(t,x). \]
Applying monotone and dominated convergence theorem, this property holds for $Q(D;t,x)$ as well, for any $D\in \mathcal{D}(t,x)$.
\end{proof}

By Theorem \ref{theorem for Q} and the definition of $Q$ as in \eqref{def of Q} we immediately obtain
\begin{equation}\label{V smaller N}
V(t,x)=\sup_{D\in \mathcal{D}(t,x)}\mathcal{J}(D;t,x)=\sup_{D\in \mathcal{D}(t,x)}\mathbb{E}\left[Q_{T-t}(D;t,x)\right]\leq N(t,x).
\end{equation}
Moreover, by definition \eqref{def of M} one has
\begin{equation}\label{M smaller V}
M(t,x) = \mathcal{J}(D^{\star}(t,x);t,x) \leq V(t,x).
\end{equation}

With all these results at hand, we can now finally prove Theorem \ref{main theorem}.

\begin{proof}[\textbf{Proof of Theorem \ref{main theorem}}]
\vspace{0.25cm}

By combining \eqref{V smaller N}, \eqref{M smaller V}, and Theorem \ref{N=M} we obtain the series of inequalities
\[ N(t,x) \geq V(t,x) \geq M(t,x)=N(t,x) \]
which proves the claim that $V=M$, and the optimality of $D^{\star}$. It just remains to prove \eqref{eq:V such that one can compute it}. To accomplish that we adapt and expand arguments as those used by El Karoui and Karatzas in the proof of Corollary 4.2 in \cite{EKK1989}.

Observe that optimality of $D^{\star}$ implies that for all $x > b(t)$
\begin{equation}\label{Representation of V eq1}
V(t,b(t))+f(t)(x-b(t))=V(t,x).
\end{equation}
Using \eqref{definition of N} and the fact that $V=N$ as proved above, we then find from \eqref{Representation of V eq1}
\begin{align*}
V(t,b(t))
&=V(t,x)-f(t)(x-b(t)) \\
&=\mathbb{E}\bigg[-\int_0^{T-t}f^{\prime}(t+s)(R_s(x)-b(t+s))^+ ~ds - \int_0^{T-t}m(t+s)~dI_s^0(x) \\
&+g(T,R_{T-t}(x))-f(t)(x-b(t))\bigg] \\
&=\mathbb{E}\bigg[-\int_0^{T-t}f^{\prime}(t+s)\Big[(R_s(x)-b(t+s))^+-(x-b(t))\Big]~ds - \int_0^{T-t}m(t+s)~dI_s^0(x) \\
&+g(T,R_{T-t}(x))-f(T)(x-b(t))\bigg].
\end{align*}

Recall \eqref{RBM without boundary}, and observe that under the condition $b(T)<\infty$ we can write
\begin{align*}
& \mathbb{E}\Big[g(T,R_{T-t}(x))\Big] = g(T,b(T)) + \mathbb{E}\bigg[\bigg(\int_{b(T)}^{R_{T-t}(x)} g_x(T,y) dy \bigg)\mathbbm{1}_{\{R_{T-t}(x) > b(T)\}} \nonumber \\
& - \bigg(\int_{R_{T-t}(x)}^{b(T)} g_x(T,y) dy \bigg)\mathbbm{1}_{\{R_{T-t}(x) \leq b(T)\}}\bigg] = g(T,b(T))  \nonumber \\
& + \mathbb{E}\bigg[f(T)\big(R_{T-t}(x)-b(T)\big)\mathbbm{1}_{\{R_{T-t}(x) > b(T)\}} - \bigg(\int_{R_{T-t}(x)}^{b(T)} g_x(T,y) dy \bigg)\mathbbm{1}_{\{R_{T-t}(x) \leq b(T)\}}\bigg],
\end{align*}
where the last equality follows from Remark \ref{rem:gxfT}. Therefore, we obtain that 
\begin{align*}
V(t,b(t))
&=\mathbb{E}\bigg[-\int_0^{T-t}f^{\prime}(t+s)\Big[(R_{s}(x)-b(t+s))^+-(x-b(t))\Big]~ds - \int_0^{T-t}m(t+s)~dI_s^0(x) \\
&+g(T,b(T))+f(T)\big(R_{T-t}(x)-b(T)\big)\mathbbm{1}_{\{R_{T-t}(x) > b(T)\}}- f(T)\big(x-b(t)\big) \nonumber \\
& - \bigg(\int_{R_{T-t}(x)}^{b(T)} g_x(T,y) dy \bigg)\mathbbm{1}_{\{R_{T-t}(x) \leq b(T)\}}\bigg].
\end{align*}
Notice now that $I_s^0(x) \to 0$, $R_s(x) \rightarrow \infty$, and $(R_s(x)-b(t+s))^+ - (x-b(t)) \rightarrow \mu s + \sigma W_s -b(t+s)+b(t)$ a.s.\ for any $s\geq0$ when $x \uparrow \infty$ (cf.\ \eqref{RBM without boundary}). Then, letting $x \to \infty$ in the last expression for $V(t,b(t))$, and invoking the monotone and dominated convergence theorems, we find (after evaluating the expectations and rearranging terms)
\begin{align*}
V(t,b(t))
&=\mathbb{E}\bigg[-\int_0^{T-t}f^{\prime}(t+s)\Big(\mu s +\sigma W_s -b(t+s)+b(t)\Big)~ds  \\
&+g(T,b(T))+f(T)\left(\mu(T-t)+\sigma W_{T-t} -b(T)+b(t)\right)\bigg] \\
%&=-\mu \int_0^{T-t}f^{\prime}(t+s)s ~ds + \int_0^{T-t}f^{\prime}(t+s)b(t+s) ~ds -f(T)b(t)+f(t)b(t)\\
%&+g(T,b(T))+f(T)\mu(T-t) -f(T)b(T)+f(T)b(t) \\
&=-\mu \int_0^{T-t}f^{\prime}(t+s)s ~ds + \int_0^{T-t}f^{\prime}(t+s)b(t+s) ~ds \\
&+g(T,b(T))+f(T)\mu(T-t)+ f(t)b(t) -f(T)b(T). 
\end{align*}
\end{proof}

\begin{remark}
As a byproduct of the fact that $V=N$ and of Lemma \ref{Lemma for Ito for N}, we have that $V \in C^{1,2}([0,T)\times (0,\infty)) \cap C^0([0,T]\times \mathbb{R}_+)$. Moreover, from \eqref{eq:main} and \eqref{Stopping problem} we have that $V$ satisfies the Neumann boundary condition $V_x(t,0)=m(t)$ for all $t \in [0,T]$.
\end{remark}
\begin{remark}
The pathwise approach followed in this section seems to suggest that some of the intermediate results needed to prove Theorem \ref{main theorem} remain valid also in a more general setting in which profits and costs in \eqref{value function of ocp} are discounted at a stochastic rate. We leave the analysis of this interesting problem for future work. 
\end{remark}

%%%%%%%%%%%%%%%%%%%%%%%%%%%%%%%%%%%%%%%%%%%%%%%%%%%%%%%%%%%%%%%%%%%%%%%%%%%%%%%%%%%%

\section{Verifying Assumption \ref{Assumption on regions and stopping time}:\\ a Case Study with Discounted Constant Marginal Profits and Costs}
\label{sec:casestudy}

In this section we consider the optimal dividend problem with capital injections
\begin{align}
\label{ex:1}
\widehat{V}(t,x)
&:= \sup_{D \in \mathcal{D}(t,x)} \mathbb{E}\left[ \int_0^{T-t}\eta e^{-rs}~dD_s - \int_0^{T-t}\kappa e^{-rs}~dI^{D}_s + \eta e^{-r(T-t)} X_{T-t}^D(x) \right] \\
%&=e^{rt} \sup_{D \in \mathcal{D}(t,x)} \mathbb{E}\left[ \int_0^{T-t}\eta e^{-r(t+s)}~dD_s - \int_0^{T-t}\kappa e^{-r(t+s)}~dI^{D}_s + \eta e^{-rT} X_{T-t}^D(x) \right]  \\
&=e^{rt} V(t,x), \nonumber
\end{align}
where we have defined
\begin{equation}
\label{control problem in the example}
V(t,x):= \sup_{D \in \mathcal{D}(t,x)} \mathbb{E}\left[\int_0^{T-t}\eta e^{-r(t+s)}~dD_s - \int_0^{T-t}\kappa e^{-r(t+s)}~dI^{D}_s + \eta e^{-rT} X_{T-t}^D(x) \right]. 
\end{equation}
It is clear from \eqref{control problem in the example} and \eqref{functional of the optimal control problem} that such a problem can be accommodated in our general setting \eqref{value function of ocp} by taking (cf.\ Assumption \ref{basic assumptions})
\begin{equation}
\label{ex:data}
f(t)=\eta e^{-rt},\quad m(t)=\kappa e^{-rt},\quad g(t,x)=\eta e^{-rt}x,
\end{equation}
for some $\kappa > \eta$ (see also Remark \ref{rem:discountedform}).

In $\widehat{V}$ of \eqref{ex:1} the coefficient $\kappa$ can be seen as a constant proportional administration cost for capital injections. On the other hand, if we immagine that transaction costs or taxes have to be paid on dividends, the coefficient $\eta$ measures a constant net proportion of leakages from the surplus received by the shareholders.

\begin{remark}
\label{rem:dividendproblem}
Problem \eqref{ex:1} is perhaps the most common formulation of the optimal dividend problem with capital injections (see, e.g., Kulenko and Schmidli \cite{KulenkoSchmidli}, Lokka and Zervos \cite{LOKKA2008954}, Zhu and Yang \cite{Yang} and references therein). However, to the best of our knowledge, no previous work has considered such a problem in the case of a finite time horizon, whereas problem \eqref{ex:1} has been extensively studied when $T=+\infty$ (see, e.g., Ferrari \cite{Ferrari17} and references therein). In particular, it has been shown, e.g., in \cite{Ferrari17} that in the case $T=+\infty$ the optimal dividend strategy is triggered by a boundary $b_{\infty}>0$ that can be characterized as the solution to a nonlinear algebraic equation (see Proposition 3.2 in \cite{Ferrari17}). In Proposition 3.6 of \cite{Ferrari17} such a trigger value is also shown to be the optimal stopping boundary of problem \eqref{OST example} below (when the optimization is performed over all the $\mathbb{F}$-stopping times).
\end{remark}

Thanks to Theorem \ref{main theorem} we know that, whenever Assumption \ref{Assumption on regions and stopping time} is satisfied, the optimal control $D^{\star}$ for problem \eqref{control problem in the example} is triggered by the optimal stopping boundary $b$ of the optimal stopping problem

\begin{align}
u(t,x) & = \sup_{\tau \in \Lambda(T-t)}\mathbb{E}\Big[e^{-r \tau}\eta \mathbbm{1}_{\{ \tau < S(x) \}} + e^{-r S(x)}\kappa \mathbbm{1}_{\{ \tau \geq S(x) \}} \Big] \nonumber \\
& =\sup_{\tau \in \Lambda(T-t)}\mathbb{E}\Big[e^{-r \tau}\eta \mathbbm{1}_{\left\{A_{\tau}(x)>0  \right\}} + e^{-r S(x)}\kappa \mathbbm{1}_{\left\{A_{\tau}(x)\leq 0  \right\}} \Big]. \label{OST example}
\end{align}

In the following we study optimal stopping problem \eqref{OST example} and verify the requirements of Assumption \ref{Assumption on regions and stopping time}.

\smallskip

Moreover, by taking the sub-optimal stopping time $\tau=0$ in \eqref{OST example} clearly gives $u(t,x)\geq \eta$ for $(t,x)\in [0,T] \times (0,\infty)$. Therefore, we can define the continuation and the stopping region of problem \eqref{OST example} as
$$\mathcal{C}:=\{(t,x) \in [0,T) \times (0, \infty): u(t,x)> \eta \}, \quad  \mathcal{S}:=\{(t,x) \in [0,T] \times (0, \infty): u(t,x)= \eta \}.$$ 
Also, notice that we have $u(t,x) \leq \kappa$ for $(t,x)\in [0,T] \times \mathbb{R}_+$ since $\eta < \kappa$.
 
Since the reward process $\phi_t:=e^{-r t}\eta \mathbbm{1}_{\{ t < S(x) \}} + e^{-r S(x)}\kappa \mathbbm{1}_{\{ t \geq S(x) \}}$ is upper semicontinuous in expectation along stopping times (thanks to the fact that $\eta<\kappa$), Theorem 2.9 in Kobylanski and Quenez \cite{Kobylanski} ensures that the first time the value process (i.e.\ the Snell envelope of the reward process) equals the reward process is optimal. In our Markovian setting we thus have that the stopping time 
\begin{equation}
\label{OST-casestudy}
\tau^{\star}(t,x):=\inf\{s\in[0,T-t):\, (t+s,A_s(x)) \in \mathcal{S}\}\wedge (T-t), \quad (t,x) \in [0,T] \times \mathbb{R}_+,
\end{equation}
is optimal. Further, defining $Z_s(x):=x + \mu s + \sigma W_s$, $s \geq 0$, the process
\begin{equation} 
\label{martingale in example}
e^{-r( s \wedge \tau^{\star}(t,x) \wedge S(x) )} u(t+( s \wedge \tau^{\star}(t,x) \wedge S(x) ),Z_{( s \wedge \tau^{\star}(t,x) \wedge S(x) )}(x)),~~ s \in [0,T-t],
\end{equation}
is an $\mathbb{F}$-martingale (cf.\ Proposition 1.6 and Remark 1.7 in Kobylanski and Quenez \cite{Kobylanski}).

The next proposition proves some preliminary properties of $u$.

\begin{proposition}
\label{properties of u in example}
The value function $u$ of \eqref{OST example} satisfies the following:
\begin{itemize}
\item[(i)] 
$u(T,x)=\eta$ for any $x>0$ and $u(t,0)=\kappa$ for any $t \in [0,T]$;
\item[(ii)]
$ t \mapsto u(t,x)$ is nonincreasing for any $x>0$;
\item[(iii)]
$x \mapsto u(t,x)$ is nonincreasing for any $t \in [0,T]$.
\end{itemize}
\end{proposition}

\begin{proof}
We prove each item separately.
\vspace{0.25cm}
\newline
(i) The first property easily follows from definition \eqref{OST example}.
\vspace{0.15cm}
\newline
(ii) The second property is due to the fact that $\Lambda(T-\cdot)$ shrinks and the expected value on the right-hand side of \eqref{OST example} is independent of $t \in [0,T]$.
\vspace{0.15cm}
\newline
(iii) Fix $t \in [0,T]$, $x_2>x_1 \geq 0$ and notice that $S(x_2) > S(x_1)$. Then, from \eqref{OST example} we can write
\begin{align}
&u(t,x_2)-u(t,x_1) \nonumber \\ 
&\leq \sup_{\tau \in \Lambda(T-t)} \mathbb{E}\bigg[ e^{-r\tau} \eta \mathbbm{1}_{\{ \tau < S(x_2) \}} -e^{-r\tau} \eta \mathbbm{1}_{\{ \tau < S(x_1) \}} 
+e^{-rS(x_2)} \kappa \mathbbm{1}_{\{ \tau \geq S(x_2) \}} -e^{-rS(x_1)} \kappa \mathbbm{1}_{\{ \tau \geq S(x_1) \}} \bigg] \nonumber \\
&= \sup_{\tau \in \Lambda(T-t)} \mathbb{E}\bigg[ \mathbbm{1}_{\{ S(x_1) \leq \tau < S(x_2) \}}  \left(e^{-r\tau} \eta-e^{-rS(x_1)} \kappa\right) +\left(e^{-rS(x_2)}-e^{-rS(x_1)}\right) \kappa \mathbbm{1}_{\{ \tau \geq S(x_2) \}} \bigg] \nonumber \\
&\leq \sup_{\tau \in \Lambda(T-t)} \mathbb{E}\bigg[ e^{-rS(x_1)} (\eta-\kappa) \mathbbm{1}_{\{ S(x_1) \leq \tau < S(x_2) \}} + \left(e^{-rS(x_2)}-e^{-rS(x_1)}\right) \kappa \mathbbm{1}_{\{ \tau \geq S(x_2) \}} \bigg] \nonumber \leq 0,
\end{align}
where we have used that $\eta < \kappa$ in the last step.
\end{proof}

Since $x \mapsto u(t,x)$ is nonincreasing for each $t \in [0,T]$, setting 
\begin{equation}
\label{Boundary}
b(t):=\inf\{x >0: u(t,x)\leq\eta\},\quad t \in [0,T], 
\end{equation}
it is clear that 
\begin{equation}
\mathcal{C}=\left\{ (t,x) \in [0,T) \times [0, \infty): 0 < x < b(t) \right\}, \quad \mathcal{S}=\left\{ (t,x) \in [0,T] \times [0, \infty): x \geq b(t) \right\}.
\end{equation}
Moreover, the optimal stopping time of \eqref{OST-casestudy} reads 
\begin{equation}
\label{OST-casestudy-2}
\tau^{\star}(t,x):=\inf\{s\in[0,T-t):\, A_s(x) \geq b(t+s)\}\wedge (T-t).
\end{equation}

In the following we will refer to $b$ as to the \emph{free boundary}. The next theorem proves preliminary properties of $b$.

\begin{proposition}
\label{main theorem example-0}
The free boundary $b$ is such that
\begin{itemize}
\item[(i)]
$t \mapsto b(t)$ is nonincreasing;
\item[(ii)] One has $b(t)>0$ for all $t\in [0,T)$. Moreover, there exists $b_{\infty} > 0$ such that $b(t) \leq b_{\infty}$ for any $t\in [0,T]$.
\end{itemize} 
\end{proposition}
\begin{proof}
We prove each item separately.
\vspace{0.15cm}

(i) The claimed monotonicity of $b$ immediately follows from (ii) of Proposition \ref{properties of u in example}.
\vspace{0.15cm}

(ii) To show that $b(t)>0$ for any $t\in [0,T)$ it is enough to observe that $u(t,0) = \kappa > \eta$ for all $t\in [0,T)$. 

To prove $b(t) < \infty$ notice that $u(t,x)\leq u_{\infty}(x)$ for all $(t,x) \in [0,T] \times \mathbb{R}_+$, where 
\[u_{\infty}(x):=\sup_{\tau \geq 0} \mathbb{E}\left[ \eta e^{-r \tau} \mathbbm{1}_{\{ \tau < S(x) \}} + \kappa e^{-r S(x)} \mathbbm{1}_{\{ \tau \geq S(x) \}}  \right]. \] Hence, setting $b_{\infty}:=\inf\{ x > 0: u_{\infty}(x)=\eta\}$ (which exists finite, e.g., by Proposition 3.2 in Ferrari \cite{Ferrari17}; see also Remark \ref{rem:dividendproblem} above), we have $b(t) \leq b_{\infty}$ for all $t \in [0,T]$.
\vspace{0.15cm}
\end{proof}

The proof of the next proposition is quite lenghty, and it is therefore postponed in the Appendix in order to simplify the exposition.
\begin{proposition}
\label{prop:lsc}
The function $(t,x) \mapsto u(t,x)$ is lower semicontinuous on $[0,T)\times (0,\infty)$.
\end{proposition}

The lower semicontinuity of $u$ implies that the martingale of \eqref{martingale in example} has right-continuous sample paths, and that the stopping region is closed. The latter fact in turn plays an important role when proving continuity of the free boundary, as it is shown in the next proposition.
\begin{proposition}
\label{main theorem example}
The free boundary $b$ is such that $t \mapsto b(t)$ is continuous on $[0,T)$. Moreover, $b(T):=\lim_{t \uparrow T} b(t)=0$.
\end{proposition}

\begin{proof}
We prove the two properties separately.
\vspace{0.25cm}

Here we show that $b$ is continuous, and this proof is divided in two parts. We start with the right-continuity.
Note that, by lower semicontinuity of $u$ (cf.\ Proposition \ref{prop:lsc}), the stopping region $\mathcal{S}$ is closed. Then fix an arbitrary point $t\in [0,T)$, take any sequence $(t_n)_{n\geq 1}$ such that $t_n \downarrow t$, and notice that $(t_n,b(t_n)) \in \mathcal{S}$, by definition. Setting $b(t+):= \lim_{t_n \downarrow t} b(t_n)$ (which exists due to Proposition \ref{main theorem example-0}-(i)), we have $(t_n,b(t_n)) \to (t,b(t+))$, and since $\mathcal{S}$ is closed $(t,b(t+)) \in \mathcal{S}$. Therefore, it holds $b(t+) \geq b(t)$ by definition \eqref{Boundary} of $b$. However, $b(\cdot)$ is nonincreasing, and therefore $b(t) = b(t+)$. 

Next we show left-continuity for all $t \in (0,T)$ and for this we adapt to our setting ideas as those in the proof of Proposition 4.2\ in De Angelis and Ekstr\"om \cite{DeAngelis}. Suppose that $b$ makes a jump at some $t \in (0,T)$. By Proposition \ref{main theorem example-0}-(i) we have $ \lim_{t_n \uparrow t} b(t_n):=b(t-) \geq b(t)$. We employ a contradiction scheme to show $b(t-) = b(t)$, and we assume $b(t-) > b(t)$. Let $x:=\frac{b(t-)+b(t)}{2}$, recall $Z_s(x)=x + \mu s + \sigma W_s$, $s \geq 0$, and define
\[ \tau_{\varepsilon} := \inf\{ s \geq 0: Z_s(x) \notin (b(t-),b(t)) \} \wedge \varepsilon \]
for $\varepsilon \in (0,t)$. Then noticing that $\tau_{\varepsilon} < \tau^{\star}(t-\varepsilon,x) \wedge S(x)$, by the martingale property of \eqref{martingale in example} we can write
\begin{align*}
u(t-\varepsilon,x)
&=\mathbb{E}\left[ e^{-r \tau_{\varepsilon}} u(t-\varepsilon+\tau_{\varepsilon},Z_{\tau_{\varepsilon}}(x)) \right] \\
&=\mathbb{E}\left[ e^{-r \varepsilon} u(t,Z_{\varepsilon}(x)) \mathbbm{1}_{\{ \tau_{\varepsilon} = \varepsilon \}} +  e^{-r \tau_{\varepsilon}} u(t-\varepsilon+\tau_{\varepsilon},Z_{\tau_{\varepsilon}}(x)) \mathbbm{1}_{\{ \tau_{\varepsilon} < \varepsilon \}} \right] \nonumber \\
&\leq \mathbb{E} \left[ e^{-r \varepsilon} \eta \mathbbm{1}_{\{ \tau_{\varepsilon} = \varepsilon \}} +  e^{-r \tau_{\varepsilon}} \kappa \mathbbm{1}_{\{ \tau_{\varepsilon} < \varepsilon \}} \right] \\
&\leq e^{-r \varepsilon} \eta + \kappa \mathbb{P}\left(\tau_{\varepsilon} < \varepsilon \right),
\end{align*}
where the last step follows from the fact that $u \leq \kappa$, and that $Z_{\tau_{\varepsilon}}(x) \geq b(t)$ on the set $\{\tau_{\varepsilon} = \varepsilon \}$. Since $e^{-r \varepsilon} \eta + \kappa \mathbb{P}(\tau_{\varepsilon} < \varepsilon)=\eta (1-r\varepsilon) + \kappa o(\varepsilon)$ as $\varepsilon \downarrow 0$, we have found a contradiction to $u(t,x)\geq \eta$. Therefore, $b(t-) = b(t)$ and $b$ is continuous on $[0,T)$.
\vspace{0.15cm}

To prove the claimed limit, notice that if $b(T):=\lim_{t \uparrow T} b(t) >0$, then any point $(T,x)$ with $ x \in (0,b(T))$ belongs to $\mathcal{C}$. However, we know that $(T,x) \in \mathcal{S}$ for all $x>0$, and we thus reach a contradiction.
\end{proof}

Thanks to the previous results all the requirements of Assumption \ref{Assumption on regions and stopping time} are satisfied for problem \eqref{OST example}. Hence Theorem \ref{main theorem} holds, and one has that $V$ of \eqref{control problem in the example} and $u$ of \eqref{OST example} are such that $V_x=u$ on $[0,T]\times \mathbb{R}_+$. In particular, by \eqref{ex:1} and Theorem \ref{main theorem} we can write
$$\widehat{V}(t,x)= \widehat{V}(t,b(t)) - e^{rt}\int_x^{b(t)}u(t,y)\,dy,$$
where by \eqref{eq:V such that one can compute it}, \eqref{ex:data}, and the fact that $b(T)=0$ we have 
$$\widehat{V}(t,b(t)) = \eta b(t) + \frac{\mu \eta}{r} \big(1 - e^{-r(T-t)}\big) - r\eta \int_t^T e^{-r(u-t)}b(u) du.$$
Moreover, the optimal dividend distributions' policy $D^{\star}$ given through \eqref{system for I and D} is triggered by the free boundary $b$ whose properties have been derived in Theorem \ref{main theorem example}.

%%%%%%%%%%%%%%%%%%%%%%%%

\subsection{A Comparative Statics Analysis.}
\label{subsec:CS}

We conclude by providing the monotonicity of the free boundary with respect to some of the problem's parameters. In the following, for any given and fixed $t\in[0,T]$, we write $b(t;\cdot)$ in order to stress the dependence of the free boundary point $b(t)$ with respect to a given parameter. Similarly, we write $u(t,x;\cdot)$ when we need to consider the dependence of $u(t,x)$, $(t,x) \in [0,T] \times \mathbb{R}_+$, with respect to a given problem's parameter.
 
\begin{proposition}
\label{prop:compstat}
Let $t\in[0,T]$ be given and fixed. It holds that
\begin{itemize}
\item[(i)] $\kappa \mapsto b(t;\kappa)$ is nondecreasing;
\item[(ii)] $\eta \mapsto b(t;\eta)$ is nonincreasing;
\item[(iii)] $r \mapsto b(t;r)$ is nonincreasing;
\item[(iv)] $\mu \mapsto b(t;\mu)$ is nonincreasing.
\end{itemize}
\end{proposition}
\begin{proof}
Recalling that
$$u(t,x) = \sup_{\tau \in \Lambda(T-t)}\mathbb{E}\Big[e^{-r \tau}\eta \mathbbm{1}_{\{ \tau < S(x) \}} + e^{-r S(x)}\kappa \mathbbm{1}_{\{ \tau \geq S(x) \}} \Big], \quad (t,x) \in [0,T] \times \mathbb{R}_+,$$
one can easily show that 
\begin{itemize}
\item[(1)] $\kappa \mapsto u(t,x;\kappa)$ is nondecreasing,
\item[(2)] $\eta \mapsto u(t,x;\eta)-\eta=\sup_{\tau \in \Lambda(T-t)}\mathbb{E}\Big[\eta \big(e^{-r \tau}\mathbbm{1}_{\{ \tau < S(x) \}}-1\big) + e^{-r S(x)}\kappa \mathbbm{1}_{\{ \tau \geq S(x) \}} \Big]$ is nonincreasing,
\item[(3)] $r \mapsto u(t,x;r)$ is nonincreasing.
\end{itemize}

Moreover, let $\mu_2 > \mu_1$ and denote by $S(x;\mu_2)$ (resp.\ $S(x;\mu_1)$) the hitting time of the origin of the drifted Brownian Motion with drift $\mu_2$ (resp.\ $\mu_1$). Since $S(x;\mu_2)\geq S(x;\mu_1)$ a.s.\ we obtain 
\begin{align*}
&u(t,x;\mu_2)-u(t,x;\mu_1)
\leq\sup_{\tau \in \Lambda (T-t)} \mathbb{E}\Big[e^{-r \tau}\eta \left(\mathbbm{1}_{\{ \tau < S(x;\mu_2) \}}-\mathbbm{1}_{\{ \tau < S(x;\mu_1) \}}\right) \nonumber   \\
&\hspace{2cm} + \kappa \left(e^{-r S(x;\mu_2)} \mathbbm{1}_{\{ \tau \geq S(x,\mu_2) \}}-e^{-r S(x;\mu_1)} \mathbbm{1}_{\{ \tau \geq S(x;\mu_1) \}}\right) \Big] \nonumber \\
&\leq\sup_{\tau \in \Lambda (T-t)} \mathbb{E}\Big[e^{-r \tau}\eta \mathbbm{1}_{\{ S(x,\mu_1) \leq \tau < S(x;\mu_2) \}}  - \kappa e^{-r S(x;\mu_1)} \mathbbm{1}_{\{ S(x,\mu_2)> \tau \geq S(x,\mu_1) \}} \nonumber \\
& \hspace{2cm} + \kappa \mathbbm{1}_{\{ \tau \geq S(x;\mu_2) \}} \left(e^{-r S(x;\mu_2)}-e^{-r S(x;\mu_1)}\right) \Big] \nonumber \\
&=\sup_{\tau \in \Lambda (T-t)} \mathbb{E}\Big[  \mathbbm{1}_{\{ S(x,\mu_1) \leq \tau < S(x;\mu_2) \}} \left( e^{-r \tau}\eta-e^{-r S(x;\mu_1)} \kappa \right) + \mathbbm{1}_{\{ \tau \geq S(x;\mu_2) \}} \left(e^{-r S(x;\mu_2)}-e^{-r S(x;\mu_1)}\right)\Big] \nonumber \\
& \leq 0. \nonumber
\end{align*}

Given the previous monotonicity properties of $u$, we can now prove items (i)-(iv).

\begin{itemize}
\item[(i)]
Taking $\kappa_2 > \kappa_1$ and using (1) and \eqref{Boundary} we have
\[ b(t;\kappa_2):=\inf\{ x >0:u(t,x;\kappa_2) \leq \eta\}\geq \inf\{ x >0:u(t,x;\kappa_1) \leq \eta\}=b(t;\kappa_1). \]
\item[(ii)]
Taking $\eta_2 > \eta_1$ and using (2) and \eqref{Boundary} we have
\[ b(t;\eta_2):=\inf\{ x >0:u(t,x;\eta_2)-\eta_2 \leq 0\}\leq \inf\{ x >0:u(t,x;\eta_1)-\eta_1 \leq 0\}=b(t;\eta_1). \]
\item[(iii)]
Taking $r_2 > r_1$ and using (3) and \eqref{Boundary} we have
\[ b(t;r_2):=\inf\{ x >0:u(t,x;r_2) \leq \eta\}\leq \inf\{ x >0:u(t,x;r_1) \leq \eta\}=b(t;r_1). \]
\item[(iv)] 
Taking $\mu_2 > \mu_1$ and that $u(t,x;\mu_2)-u(t,x;\mu_1) \leq 0$ and \eqref{Boundary} we have
\[ b(t;\mu_2):=\inf\{ x >0:u(t,x;\mu_2) \leq \eta\}\leq \inf\{ x >0:u(t,x;\mu_1) \leq \eta\}=b(t;\mu_1). \]
\end{itemize}
\end{proof}

The last proposition allows us to draw some economic implications. Increasing the parameters $\eta$, $r$, and $\mu$, leads, at each time $t$, to an earlier dividends' distribution. This result is quite intuitive since an higher interest rate $r$ lowers future profits due to discounting, an higher $\eta$ increases the marginal value of dividends, and an higher $\mu$ increases the surplus' trend and lowers the probability of bankruptcy, hence of capital injections. On the other hand, an increase of $\kappa$ postpones the dividends' distribution since capital injections become more expensive, and the fund's manager thus acts in a more cautious way. 

Proving the monotonicity of the free boundary with respect to the surplus' volatility $\sigma$ seems not to be feasible by following the arguments of the proof of Proposition \ref{prop:compstat}. One should then rely on a careful numerical analysis of the dynamic programming equation associated to the optimal dividend problem, and we believe that such a study falls outside the scopes of this work. However, we conjecture that an increase of $\sigma$ should postpone the dividends' distribution. Indeed, the larger $\sigma$ is, the higher becomes the risk of the need of costly capital injections. As a consequence, the fund's manager wants to wait longer before distributing an additional unit of dividends. Such a monotonicity of the free boundary with respect to $\sigma$ has been recently proved by Ferrari in Proposition 4.1 of \cite{Ferrari17} in the case of a stationary optimal dividend problem with capital injections.

%%%%%%%%%%%%%%%%%%%%%%%%%%%%%%%%%%%%%%%%%%%%%%%%

\section*{Acknowledgments}
\noindent Financial support by the German Research Foundation (DFG) through the Collaborative Research Centre 1283 ``Taming uncertainty and profiting from randomness and low regularity in analysis, stochastics and their applications'' is gratefully acknowledged. 

Both the authors thank Miryana Grigorova and Hanspeter Schmidli for fruitful discussions and comments. Part of this work has been finalized while the first author was visiting the Department of Mathematics of the University of Padova thanks to the program ``Visiting Scientists 2018''. Giorgio Ferrari acknowledges the Department of Mathematics of the University of Padova for the hospitality. We also wish to thank three anonymous referees for their careful reading and inspiring comments.

%%%%%%%%%%%%%%%%%%%%%%%%%%%%%%%%%%%%%%%%%%%%%%%%%%%%%%%

\newpage

\appendix

\section{Appendix}

\subsection{Proof of Corollary \ref{smooth fit for u}}

Notice that from \eqref{second representation of u} we can write for any $x>0$ and $t \in [0,T]$

\begin{align}
u(t,x)
&=\mathbb{E}\bigg[ \int_0^{T-t} -f^{\prime}(t+\theta)\mathbbm{1}_{\{x + \mu \theta + \sigma W_{\theta} \geq b(t+\theta)\}} \mathbbm{1}_{\{\theta < S(x)\}} ~d\theta  \nonumber \\
&+ m(t+S(x))\mathbbm{1}_{\{S(x) \leq T-t\}}+g_x(T,A_{T-t}(x)) \bigg] \nonumber \\
&= \int_0^{T-t} -f^{\prime}(t+\theta) \mathbb{P}\big(x + \mu \theta + \sigma W_{\theta} \geq b(t+\theta), S(x) > \theta \big)~d\theta    \label{u in terms of probabilities}\\
&+ \mathbb{E}\big[ m(t+S(x))\mathbbm{1}_{\{S(x) \leq T-t\}}\big] + \mathbb{E}\big[ g_x(T,A_{T-t}(x)) \big], \nonumber
\end{align}
where Fubini's theorem and the fact that $f^{\prime}$ is deterministic has been used for the integral term above.

We now investigate the three summands separately.
By using Proposition 3.2.1.1 in Jeanblanc et al.\ \cite{jeanblanc2009mathematical}, and recalling that the stopping boundary $b$ is strictly positive by Assumption \ref{Assumption on regions and stopping time}, we have
\begin{align}
&\mathbb{P}\bigg(x + \mu \theta + \sigma W_{\theta} \geq b(t+\theta), S(x) > \theta \bigg) \nonumber \\
&=\mathbb{P}\bigg(x + \mu \theta + \sigma W_{\theta} \geq b(t+\theta), \inf_{s \leq \theta} (x + \mu s + \sigma W_{s}) > 0 \bigg) \nonumber \\
&=\mathbb{P}\bigg( \frac{\mu}{\sigma} \theta +  W_{\theta} \geq \frac{b(t+\theta)-x}{\sigma}, \inf_{s \leq \theta} \left( \frac{\mu}{\sigma} s + W_{s} \right) > -\frac{x}{\sigma} \bigg) \label{first density for u} \\
&= \mathcal{N}\bigg( \frac{\frac{x-b(t+\theta)}{\sigma}+\frac{\mu}{\sigma}\theta}{\sqrt{\theta}}\bigg) - e^{-2\frac{\mu x}{\sigma^2} } \mathcal{N}\bigg( \frac{-\frac{b(t+\theta)+x}{\sigma}+\frac{\mu}{\sigma}\theta}{\sqrt{\theta}}\bigg). \nonumber
\end{align}
Here $\mathcal{N}(\,\cdot\,)$ denotes the cumulative distribution function of a standard Gaussian random variable. Note that the last term in \eqref{first density for u} is continuously differentiable with respect to $x$ for any $\theta >0$.

For the second summand in the last expression on the right-hand side of \eqref{u in terms of probabilities} we first rewrite $S(x)$, for $x \geq 0$, as
\begin{align}
S(x)
&= \inf\{s \geq 0: x+ \mu s + \sigma W_s =0\} = \inf\{s \geq 0:  \frac{\mu}{\sigma} s + W_s =-\frac{x}{\sigma}\} \nonumber \\
&\overset{\mathcal{L}}{=} \inf\{s \geq 0:  -\frac{\mu}{\sigma} s + \widehat{W}_s =\frac{x}{\sigma}\} \label{rewritten S(x)}.
\end{align}
where $\widehat{W}$ is a standard Brownian motion. Hence equation $(3.2.3)$ in Jeanblanc et al.\ \cite{jeanblanc2009mathematical} applies and allows us to write the probability density of $S(x)$ as
\begin{equation}
\rho_{S(x)}(u):=\frac{d\mathbb{P}(S(x) \in du)}{du}= \frac{x}{\sigma\sqrt{2 \pi u^3}}e^{-\frac{(\frac{x}{\sigma}+\frac{\mu}{\sigma}u)^2}{2u}}, \quad u \geq 0. \label{density of S(x)}
\end{equation}

For the third summand we notice that the absorbed process $A_{T-t}(x)$ of \eqref{absorbed process} is the drifted Brownian motion started in $x$ and killed at the origin. Denote by $\rho_{A}(t,x,y)$ its transition density of moving from $x$ to $y$ in $t$ units of time. Then, by employing the result of Borodin and Salminen \cite{boroding2002handbook}, Section 15 in Appendix 1 (suitably adjusted to our case with $\sigma \neq 1$), we obtain
\begin{align}
\label{density of the absorbed process}
\rho_{A}\left(T-t,x,y\right)
&:=\frac{d\mathbb{P}(A_{T-t}(x) \in dy)}{dy}
=\frac{1}{\sqrt{2 \pi (T-t)\sigma^2}} \exp\left( -\left(\frac{\mu(x-y)}{\sigma^2}\right)-\frac{\mu^2}{2 \sigma^2}(T-t) \right) \nonumber \\
&\times \left( \exp\left( -\frac{(x-y)^2}{2\sigma^2 (T-t)} \right)- \exp\left( -\frac{(x+y)^2}{2\sigma^2 (T-t)} \right) \right). 
\end{align}

Feeding \eqref{first density for u}, \eqref{density of S(x)} and \eqref{density of the absorbed process} back into \eqref{u in terms of probabilities} we obtain
\begin{align}
u(t,x)
&=\int_0^{T-t} -f^{\prime}(t+\theta) \bigg[ \mathcal{N}\bigg( \frac{\frac{x-b(t+\theta)}{\sigma}+\frac{\mu}{\sigma}\theta}{\sqrt{\theta}}\bigg) - e^{-2\frac{\mu x}{\sigma^2} } \mathcal{N}\bigg( \frac{-\frac{b(t+\theta)+x}{\sigma}+\frac{\mu}{\sigma}\theta}{\sqrt{\theta}}\bigg) \bigg]~d\theta \nonumber \\
& + \int_0^{T-t} m(t+u)\rho_{S(x)}(u)~du \label{u in form of densities} +\int_0^{\infty} g_x(T,y) \rho_{A}\left(T-t,x,y\right) ~dy, 
\end{align}
and it is easy to see by the dominated convergence theorem that $x \mapsto u(t,x)$ is continuously differentiable on $(0,\infty)$ for any $t<T$.

\subsection{Proof of Lemma \ref{Lemma for Ito for N}}

By \eqref{integration of u with N} and Corollary \ref{smooth fit for u} the function $N$ of \eqref{definition of N} is twice-continuously differentiable with respect to $x$ on $(0,\infty)$.
To show that $N$ is also continuously differentiable with respect to $t$ on $[0,T)$ we express the expected value on the right-hand side of \eqref{definition of N} as an integral with respect to the probability densities of the involved processes. We thus start computing the transition density of the reflected Brownian motion $R$ of \eqref{RBM 2}, which we call $\rho_{R}$.
By Appendix 1, Chapter 14, in Borodin and Salminen \cite{boroding2002handbook} (easily adapted to our case with $\sigma \neq 1$) we have
\begin{align}
\rho_{R}(u,x,y)
&:=\frac{d\mathbb{P}(R_u(x) \in dy)}{dy}
=\frac{1}{\sqrt{2 \pi u \sigma^2}} \exp\left( -\frac{\mu}{\sigma}\left(\frac{x-y}{\sigma}\right)-\frac{\mu^2}{2 \sigma^2} u \right) \times \nonumber \\
&\left( \exp\left( -\frac{(x-y)^2}{2\sigma^2 u} \right)- \exp\left( -\frac{(x+y)^2}{2\sigma^2 u} \right) \right) - \frac{\mu}{2 \sigma} \text{Erfc}\left(\frac{x+y+\mu u}{\sqrt{2\sigma ^2 u}}\right),
\end{align}
where Erfc$(x):=\int_{-\infty}^x \frac{1}{\sqrt{2\pi}} e^{-\frac{y^2}{2}}~dy$ for $x \in \mathbb{R}$.
Hence, by using Fubini's Theorem, \eqref{definition of N} reads as

\begin{align}
N(t,x)
&=\mathbb{E}\bigg[ -\int_0^{T-t} \left(R_{s}(x)-b(t+s)\right)^+ f^{\prime}(t+s)~ds - \int_0^{T-t} m(t+s)~dI^0_{s}(x) \nonumber \\
&+ g(T,R_{T-t}(x)) \bigg]= -\int_t^{T} \mathbb{E}\Big[\left(R_{u-t}(x)-b(u)\right)^+\Big] f^{\prime}(u)~du \nonumber \\
&- \mathbb{E}\bigg[\int_0^{T-t} m(t+s)~dI^0_{s}(x)\bigg] + \mathbb{E}\Big[g(T,R_{T-t}(x)) \Big] \nonumber \\
&= -\int_t^{T} \bigg(\int_0^{\infty} \left(y-b(u)\right)^+ \rho_{R}(u-t,x,y)~dy\bigg) f^{\prime}(u)~du - \mathbb{E}\bigg[\int_t^{T} m(u)~dI^0_{u-t}(x)\bigg] \label{N rewritten}  \\
&+ \int_0^{\infty}g(T,y)\rho_{R}(T-t,x,y) ~dy.  \nonumber
\end{align}

Recalling that $m$ is continuously differentiable by Assumption \ref{basic assumptions} and using an integration by parts, we can write
\begin{align}
\mathbb{E}\bigg[\int_t^{T} m(u)~dI^0_{u-t}(x)\bigg]
&=\mathbb{E}\bigg[m(T) I^0_{T-t}(x) -\int_t^{T} I^0_{u-t}(x) m^{\prime}(u) ~du\bigg] \nonumber \\
&=m(T) \mathbb{E}\big[I^0_{T-t}(x)\big] -\int_t^{T} \mathbb{E}\big[I^0_{u-t}(x)\big] m^{\prime}(u) ~du \nonumber \\
&=m(T) \mathbb{E}\big[0 \vee (\sigma \xi_{T-t}-x)\big] -\int_t^{T} \mathbb{E}\big[0 \vee (\sigma \xi_{u-t}-x)\big] m^{\prime}(u) ~du \nonumber, 
\end{align}
where we have used that $I^0_{s}(x)=0 \vee (\sigma \xi_s-x)$ with $\xi_s:=\sup_{\theta \leq s}( - \frac{\mu}{\sigma} \theta - W_{\theta})$.
%$\overset{\mathcal{L}}{=}\sup_{\theta \leq s} (-\frac{\mu}{\sigma} \theta + \widehat{W}_{\theta})$, with $\widehat{W}$ a standard Brownian motion.
Since (cf.\ Chapter 3.2.2 in Jeanblanc et al.\ \cite{jeanblanc2009mathematical})
\begin{equation}
\mathbb{P} \left( \xi_s \leq z \right)= \mathcal{N}\left( \frac{z-\frac{\mu}{\sigma}s}{\sqrt{s}} \right)-\exp \left( 2 \frac{\mu}{\sigma} z \right) \mathcal{N} \left( \frac{-z-\frac{\mu}{\sigma}s}{\sqrt{s}} \right),  \label{Law of running supremum}
\end{equation} 
we get
\begin{align}
\mathbb{E}\bigg[0 \vee (\sigma \xi_{u-t}-x)\bigg]
&=\int_{\frac{x}{\sigma}}^{\infty}(\sigma z -x)\rho_{\xi}(u-t,z)~dz,
\end{align}
where we have defined $\rho_\xi(s,z):= \frac{d\mathbb{P}(\xi_s \leq z)}{dz}$.
Because $\rho_{\xi}(\cdot,z)$ and $\rho_{R}(\cdot,x,y)$ are continuously differentiable on $(0,T]$, it follows that $N(t,x)$ as in \eqref{N rewritten} is continuously differentiable with respect to $t$, for any $t<T$. The continuity of $N$ on $[0,T] \times \mathbb{R}_+$ also follows from the previous equations.

\subsection{Proof of Proposition \ref{prop:lsc}}

Let $(t,x) \in [0,T)\times (0,\infty)$ be given and fixed, and take any sequence $(t_n,x_n) \subset [0,T)\times (0,\infty)$ such that $(t_n,x_n) \to (t,x)$. Then, let $\tau^{\star}:=\tau^{\star}(t,x)$ be the optimal stopping time for $u(t,x)$ of \eqref{OST-casestudy-2}. From \eqref{OST example} and the fact that $\tau^{\star}\leq T-t$ a.s.\ we then find 
{\allowdisplaybreaks
\begin{align}
u(t,x)-u(t_n,x_n)
&\leq \mathbb{E}\left[ \eta e^{-r\tau^{\star}} \mathbbm{1}_{\{ \tau^{\star} < S(x) \}} + \kappa e^{-rS(x)} \mathbbm{1}_{\{ \tau^{\star} \geq S(x) \}} \right. \nonumber \\
&\left. -\eta e^{-r (\tau^{\star}\wedge (T-t_n))} \mathbbm{1}_{\{ \tau^{\star}\wedge (T-t_n) < S(x_n) \}} - \kappa e^{-rS(x_n)} \mathbbm{1}_{\{ \tau^{\star}\wedge (T-t_n) \geq S(x_n) \}} \right]  \nonumber \\
&= \mathbb{E} \left[ \mathbbm{1}_{\{ \tau^{\star} \leq T-t_n \}} \bigg\{ \eta e^{-r \tau^{\star}}\left(\mathbbm{1}_{\{ \tau^{\star} \geq S(x_n) \}}- \mathbbm{1}_{\{ \tau^{\star} \geq  S(x) \}} \right) \right. \nonumber \\
&+\left. \kappa \left( e^{-rS(x)} \mathbbm{1}_{\{ \tau^{\star} \geq S(x)\}}- e^{-rS(x_n)} \mathbbm{1}_{\{ \tau^{\star} \geq S(x_n) \}} \right) \bigg\}  \right] \nonumber \\
&+\mathbb{E} \left[ \mathbbm{1}_{\{ \tau^{\star} > T-t_n \}} \bigg\{ \eta e^{-r \tau^{\star}} \mathbbm{1}_{\{ \tau^{\star} < S(x) \}}- \eta e^{-r (T-t_n)} \mathbbm{1}_{\{ T-t_n <  S(x_n) \}}  \right. \nonumber \\
&+\left. \kappa \left( e^{-rS(x)} \mathbbm{1}_{\{ \tau^{\star} \geq S(x)\}}- e^{-rS(x_n)} \mathbbm{1}_{\{ T-t_n \geq S(x_n) \}} \right) \bigg\} \right] \nonumber \\
&\leq \mathbb{E} \left[ \mathbbm{1}_{\{ \tau^{\star} \leq T-t_n \}} \bigg\{ \eta e^{-r \tau^{\star}}\mathbbm{1}_{\{ S(x_n) \leq \tau^{\star} < S(x) \}} \right. \nonumber \\
&+\left. \kappa \left( \left|e^{-rS(x)} - e^{-rS(x_n)}\right| \mathbbm{1}_{\{ \tau^{\star} \geq S(x_n) \vee S(x_n) \}} + e^{-rS(x)} \mathbbm{1}_{\{S(x_n) > \tau^{\star} \geq S(x)\}} \right) \bigg\}  \right] \nonumber \\
&+\mathbb{E} \left[ \mathbbm{1}_{\{ \tau^{\star} > T-t_n \}} \bigg\{ \eta e^{-r (T-t_n)} \left(\mathbbm{1}_{\{ T-t_n < S(x) \}}-  \mathbbm{1}_{\{ T-t_n <  S(x_n) \}} \right) \right. \nonumber \\
&+ \kappa \mathbbm{1}_{\{ T-t>S(x) \}} \left( e^{-rS(x)} \mathbbm{1}_{\{ \tau^{\star} \geq S(x)\}}- e^{-rS(x_n)} \mathbbm{1}_{\{ T-t_n \geq S(x_n) \}} \right)\nonumber \\
&+ \kappa \mathbbm{1}_{\{ T-t=S(x) \}} \left( e^{-rS(x)} \mathbbm{1}_{\{ \tau^{\star} \geq S(x)\}}- e^{-rS(x_n)} \mathbbm{1}_{\{ T-t_n \geq S(x_n) \}} \right) \nonumber \\
&\left.+ \kappa \mathbbm{1}_{\{ T-t<S(x) \}} \left( e^{-rS(x)} \mathbbm{1}_{\{ \tau^{\star} \geq S(x)\}}- e^{-rS(x_n)} \mathbbm{1}_{\{ T-t_n \geq S(x_n) \}} \right)\bigg\} \right] \nonumber \\
&\leq \mathbb{E} \left[   \eta e^{-r \tau^{\star}}\mathbbm{1}_{\{ S(x_n) \leq \tau^{\star} < S(x) \}} \right. \nonumber +\left. \kappa \left( \left|e^{-rS(x)} - e^{-rS(x_n)}\right| + \mathbbm{1}_{\{S(x_n) > \tau^{\star} \geq S(x)\}} \right)   \right] \nonumber \\
&+\mathbb{E} \left[ \mathbbm{1}_{\{ \tau^{\star} > T-t_n \}} \bigg\{ \eta e^{-r (T-t_n)} \mathbbm{1}_{\{ S(x_n) \leq T-t_n < S(x) \}} \right. \nonumber \\
&+ \kappa \mathbbm{1}_{\{ T-t>S(x) \}} \left( e^{-rS(x)} \mathbbm{1}_{\{ T-t \geq S(x)\}}- e^{-rS(x_n)} \mathbbm{1}_{\{ T-t_n \geq S(x_n) \}} \right)\nonumber \\
&+ \kappa \mathbbm{1}_{\{ T-t=S(x) \}} \left.+ \kappa \mathbbm{1}_{\{ T-t<S(x) \}}  \mathbbm{1}_{\{ \tau^{\star} \geq S(x)\}}\bigg\} \right]. \nonumber 
\end{align} 
Rearranging terms and taking limit inferior as $n \uparrow \infty$ on both sides one obtains
\begin{align}
\underline{\lim}_{n\rightarrow\infty} u(t_n,x_n) 
&\geq u(t,x) - \overline{\lim}_{n\rightarrow\infty} \mathbb{E} \bigg[  \eta e^{-r \tau^{\star}}\mathbbm{1}_{\{ S(x_n) \leq \tau^{\star} < S(x) \}}  \nonumber \\
&+ \kappa \left( \left|e^{-rS(x)} - e^{-rS(x_n)}\right|  +  \mathbbm{1}_{\{S(x_n) > \tau^{\star} \geq S(x)\}} \right)  \bigg] \nonumber \\
&- \overline{\lim}_{n\rightarrow\infty} \mathbb{E} \bigg[ \mathbbm{1}_{\{ \tau^{\star} > T-t_n \}} \bigg\{ \eta e^{-r (T-t_n)} \mathbbm{1}_{\{ S(x_n) \leq T-t_n < S(x) \}}  \nonumber \\
&+ \kappa \mathbbm{1}_{\{ T-t>S(x) \}} \left( e^{-rS(x)} \mathbbm{1}_{\{ T-t \geq S(x)\}}- e^{-rS(x_n)} \mathbbm{1}_{\{ T-t_n \geq S(x_n) \}} \right)\nonumber \\
&+ \kappa \mathbbm{1}_{\{ T-t=S(x) \}} + \kappa \mathbbm{1}_{\{S(x) \leq \tau^{\star} \leq T-t <S(x) \}} \bigg\}  \bigg]  \nonumber \\
&\geq u(t,x) -  \mathbb{E}\left[ \kappa \mathbbm{1}_{\{S(x) = \tau^{\star} \}}  \right] - \mathbb{E} \bigg[   \eta e^{-r (T-t)} \mathbbm{1}_{\{ T-t = S(x) \}} + \kappa \mathbbm{1}_{\{ T-t=S(x) \}}  \bigg] \nonumber \\
&= u(t,x) - \kappa \mathbb{P}\left(\tau^{\star} = S(x)\right) - \left(\eta e^{-r(T-t)} + \kappa\right) \mathbb{P}\left(T-t=S(x)\right). \nonumber
\end{align}}
The last inequality follows by interchanging expectations and limits by the dominated convergence theorem, using that $S(x_n) \to S(x)$, carefully investigating the involved limits superior, and observing that $\{\tau^{\star} \geq T-t \}=\{ \tau^{\star} = T-t \}$ since $\tau^{\star} \in \Lambda(T-t)$. 

Using now that $\{T-t=S(x)\}$ is a $\mathbb{P}$-null set by \eqref{density of S(x)}, and the fact that $\mathbb{P}\left(\tau^{\star} = S(x)\right)=0$ since the free boundary is strictly positive on $[0,T)$, we then obtain 
\begin{equation}
\underline{\lim}_{n\rightarrow\infty} u(t_n,x_n) \geq u(t,x),
\end{equation}
which proves the claimed lower semicontinuity of $u$ on $[0,T) \times (0,\infty)$.

\subsection{Lemma \ref{Lemma S(z) smaller T}}

\begin{lemma}\label{Lemma S(z) smaller T}
Recall that (cf.\ \eqref{z in Case 3})
\[z=\inf\left\{y \in [0,b(0)]: \tau^{\star}(0,y)<S(y)\right\}.\]
Then it holds that
\begin{equation}
S(z)\leq T \quad a.s.
\end{equation}
\end{lemma}

\begin{proof}
In order to simplify exposition, in the following we shall stress the dependence on $\omega$ only when strictly necessary. Suppose that there exists a set $\Omega_0 \subset \Omega$ s.t.\ $\mathbb{P}(\Omega_0) >0$, and that for any $\omega \in \Omega_0$ we have $S(z)>T$. Then take $\omega_0 \in \Omega_0$, recall that $Z_s(x)=x + \mu s + \sigma W_s$ for any $x>0$ and $s\geq 0$, and notice that $\min_{0 \leq s \leq T} Z_s(z;\omega_0)=\ell:=\ell(\omega_0)>0$. Then, defining $\widehat{z}(\omega_o):=\widehat{z}=z - \frac{\ell}{2}$, one has
\begin{equation*}
\min_{0 \leq s \leq T}Z_s(\widehat{z};\omega_0)=\min_{0 \leq s \leq T}\left(z + \mu s +\sigma W_s(\omega_0)-\frac{\ell}{2}\right)=\ell-\frac{\ell}{2} =\frac{\ell}{2} >0.
\end{equation*}
Hence, $S(\widehat{z})>T \geq \tau^{\star}(0,\widehat{z})$, but this contradicts the definition of $z$ since $\widehat{z}<z$. Therefore we conclude that $S(z)\leq T$ a.s.
\end{proof}

\end{document}